\documentclass[aos,preprint]{imsart}

\RequirePackage[OT1]{fontenc}
\RequirePackage{amsthm,amsmath}
\RequirePackage[numbers]{natbib}

\usepackage{fullpage}
\usepackage{amsmath}
\usepackage{color}
\usepackage{amsfonts}
\usepackage{amssymb}
\usepackage{graphicx}
\usepackage{subfig}
\usepackage{psfrag}
\usepackage{epstopdf}
\usepackage{rotating}
\usepackage{lscape}
\usepackage{geometry}
\usepackage{mathrsfs}
\usepackage{algorithm,algorithmic}
\usepackage{url}

\newtheorem{cor}{Corollary}[section]

\newtheorem{lem}{Lemma}[section]
\newtheorem{thm}{Theorem}[section]

\newtheorem{remark}{Remark}[section]
\theoremstyle{definition}

\makeatother

\allowdisplaybreaks
\theoremstyle{definition}

\theoremstyle{remark}

\newtheorem*{rem*}{Remark}

\newcommand{\xkhB}[1]{\Big(#1\Big)}
\newcommand{\zkhB}[1]{\Big[#1\Big]}
\newcommand{\dkhB}[1]{\Big\{#1\Big\}}

\newcommand{\norm}[1]{\left\Vert#1\right\Vert}
\newcommand{\abs}[1]{\left\vert#1\right\vert}
\newcommand{\normz}[1]{{\Vert#1\Vert}_0} 
\newcommand{\normo}[1]{{\Vert#1\Vert}_1} 
\newcommand{\normt}[1]{{\Vert#1\Vert}_2} 
\newcommand{\normi}[1]{{\Vert#1\Vert}_\infty}

\def\bX{X}

\def\bP{\mathbb{P}}
\def\bRp{\mathbb{R}^p}

\def\ds{d^\diamond}
\def\As{A^\diamond}
\def\Is{I^\diamond}

\def\dk{d^k}
\def\dkk{d^{k+1}}
\def\Ak{A^k}
\def\Akk{A^{k+1}}
\def\Ik{I^k}
\def\bbk{\mathbf{\beta}^{k}}
\def\bbkk{\mathbf{\beta}^{k+1}}
\def\de{\triangleq}
\def\cm{c_-} 
\def\cp{c_+} 

\def\JAk{E_2(A^k)}
\def\JAkk{E_2(A^{k+1})}
\def\JAki{E_{\infty}(A^k)}
\def\JAkki{E_{\infty}(A^{k+1})}
\def\bbb{\triangle^k}
\def\cN{\mathcal{N}}
\def\bP{\mathbf{P}}
\def\hT{\hat{T}}

\def\real{\mathbb R}
\def\hbeta{\hat{\beta}}
\def\bbs{\beta^\diamond}

\def\bb{\mathbf{\beta}}
\def\bbp{\mathbf{\beta}^{*}}

\def\tb{\bar{\beta}^*}
\def\te{\bar{\eta}}
\def\htt{h_{2}(T)}
\def\hit{h_{\infty}(T)}
\def\la{\lambda}
\def\s{\sigma}
\def\vps{\varepsilon}
\makeatletter

\newcommand{\Rmnum}[1]{\expandafter\@slowromancap\romannumeral #1@}
\makeatother

\startlocaldefs
\numberwithin{equation}{section}
\theoremstyle{plain}
\endlocaldefs

\begin{document}

\begin{frontmatter}
\title{A Constructive Approach to High-dimensional Regression} 
\runtitle{SDAR}

\begin{aug}
\author{\fnms{Jian} \snm{Huang}\thanksref{m1}\ead[label=e1]{jian-huang@uiowa.edu}},
\author{\fnms{Yuling} \snm{Jiao}\thanksref{m2}\ead[label=e2]{yulingjiaomath@whu.edu.cn}},
\author{\fnms{Yanyan} \snm{Liu}\thanksref{m3}\ead[label=e3]{liuyy@whu.edu.cn}},
\and
\author{\fnms{Xiliang} \snm{Lu}\thanksref{m3}\ead[label=e4]{xllv.math@whu.edu.cn}
\ead[label=u1,url]{http://www.foo.com}}


\affiliation{University of Iowa\thanksmark{m1}, Zhongnan University of Economics and Law\thanksmark{m2}, Wuhan University\thanksmark{m3}}

\address{Jian Huang, \\
Department of Statistics and Actuarial Science\\
241 SH \\
 University of Iowa \\
  Iowa City, Iowa 52242, USA\\
\printead{e1}}

\address{Yuling Jiao\\
School of  Statistics and Mathematics \\
Zhongnan  University of Economics and Law \\
Wuhan,  China, \\
and Institute of Big Data \\
Zhongnan  University of Economics and Law \\
Wuhan,  China. \\
\printead{e2}}

\address{Yanyan Liu\\
School of Mathematics and Statistics\\
Wuhan University \\
Whuhan, China.\\
\printead{e3}}

\address{Xiliang Lu \\
School of Mathematics and Statistics\\
Wuhan University \\
Whuhan, China,\\
and Hubei Key Laboratory of Computational Science \\
Wuhan University\\
Whuhan, China.\\
\printead{e4}}
\end{aug}

\bigskip
\begin{abstract}
We develop a constructive approach to estimating sparse, high-dimensional linear regression models. 
The approach is a computational algorithm motivated from the KKT conditions for the $\ell_0$-penalized least squares solutions. It generates a sequence of solutions iteratively, based on support detection using primal and dual information and root finding.  We refer to the algorithm as SDAR for brevity. Under a sparse Rieze condition on the design matrix and certain other conditions, we show that with high probability,
the $\ell_2$ estimation error of the solution sequence
decays exponentially to the minimax error bound in $O(\sqrt{J}\log(R))$ steps; and under a mutual coherence condition and certain other
conditions, the $\ell_{\infty}$ estimation error
decays to the optimal error bound in $O(\log(R))$ steps,
where $J$  is the number of important predictors,
$R$ is the relative magnitude of the nonzero target coefficients. Computational complexity analysis shows that the cost of SDAR is $O(np)$
per iteration.
Moreover the oracle least squares estimator can be exactly recovered with high probability at the same cost if we know the sparsity level.
We also consider 
an adaptive version of SDAR to make it more practical in applications. Numerical comparisons with 
Lasso, MCP and  greedy methods demonstrate that 
SDAR is competitive with or outperforms them in accuracy and efficiency.
\end{abstract}


\end{frontmatter}

\section{Introduction}

Consider the linear regression model
\begin{equation}\label{model}
 y = X\beta^{*} + \eta
\end{equation}
where $y\in \real^n$  is a response vector, $X\in \real^{n\times p}$ is the design matrix with $\sqrt{n}$-normalized columns,
$\beta^{*} =(\beta^{*}_{1}, \ldots ,
\beta^{*}_{p})^{\prime} \in \real^{p}$ is the vector of underlying
regression coefficients  and $\eta  \in \real^{n}$ is a vector of random errors
with mean $0$ and variance $\sigma^2$.
We focus on the case where $p \gg n$ and the model is sparse in the sense
that only a relatively small number of predictors are important.

Without any constraints on $\beta^{*}$ there exist infinitely many least squares solutions for (\ref{model}) since it is a highly undetermined linear system
when $p \gg n$.  These solutions usually over-fit the data.
Under the assumption that $\beta^{*}$ is sparse in the sense that the number
of important nonzero elements of $\beta^{*}$ 
is small relative to  $n$, we can estimate $\beta^{*}$ by the solution of the $\ell_0$ minimization problem
\begin{equation}\label{l0}
 \min_{\beta\in \mathbb{R}^{p}}\frac{1}{2n}\|X \beta
 -y\|^{2}_{2},   \quad \textrm{subject to}  \quad \|\beta\|_{0}
 \leq s,
\end{equation}
where  $s>0$  controls the sparsity level. However,
(\ref{l0}) is generally NP hard [Natarajan (1995)],  hence it is
challenging to design a stable and fast algorithm to solve it.

In this paper we propose a 
constructive approach to estimating (\ref{model}). 
The approach is a computational algorithm motivated from the
necessary KKT conditions for the Lagrangian form of (\ref{l0}).
It finds an approximate sequence of solutions to the KKT equations
iteratively using a support detection and root finding method
until convergence is achieved.
For brevity, we refer to the proposed approach as SDAR.

\subsection{Literature review}
Several approaches have been proposed to approximate (\ref{l0}).
Among them the Lasso [Tibshirani (1996),  Chen, Donoho and Saunders (1998)], which uses
the $\ell_1$ norm of $\beta$ in the constraint instead of the $\ell_0$ norm
in (\ref{l0}), is a popular method.
Under the irrepresentable condition on the design matrix $X$ and a sparsity assumption on  $\beta^{*}$, Lasso is model selection (and sign) consistent [Meinshausen and B\"{u}hlmann (2006), Zhao and Yu (2006), Wainwright (2009)].
Lasso is a convex minimization problem. Several fast algorithms have been proposed, including LARS (Homotopy)
[Osborne, Presnell and  Turlachet (2000), Efron et al. (2004),
Donoho and  Tsaig (2008)], coordinate descent [Fu (1998),
Friedman et al. (2007),  Wu and Lange (2008)], and proximal gradient descent
[Agarwal, Negahban and  Wainwright (2012), Xiao and Zhang (2013), Nesterov (2013)].


However, Lasso tends to overshrink large coefficients, which leads to biased estimates
[Fan and Li (2001), Fan and Peng (2004)]. The adaptive Lasso proposed by Zou (2006) and analyzed by Huang, Ma and Zhang (2008) in high-dimensions can
achieve the oracle property under certain conditions.  But its requirements on the minimum value of the nonzero coefficients are not optimal.
Nonconvex penalties such as
the smoothly clipped absolute deviation (SCAD) penalty [Fan and Li (2001)],
the minimax concave penalty (MCP) [Zhang (2010a)],
and the capped $\ell_{1}$ penalty [Zhang (2010b)]  were adopted to remedy
these problems.
Although the global minimizers (also there exist some local minimizers) of these  nonconvex regularized models can eliminate the estimation
bias and enjoy the oracle properties [Zhang and Zhang (2012)], computing their global minimizers or local minimizers with the desired statistical properties is challenging since the optimization problem is nonconvex, nonsmooth and large scale in general.

There are several numerical algorithms for nonconvex regularized problems. The first kind of such methods can be considered a special case (or variant) of minimization-maximization algorithm [Lange,  Hunter  and   Yang (2000),  Hunter and Li (2005)] or of multi-stage convex relaxation [Zhang (2010b)]. Examples include
local quadratic approximation (LQA) [Fan and Li (2001)], local linear approximation (LLA) [Zou and Li (2008)], decomposing the penalty into a difference of two convex terms  (CCCP) [Kim, Cho and Oh (2008),  Gasso, Rakotomamonjy and Canu (2009)].
The second type of methods is the 
coordinate descent  algorithms,  including coordinate descent of the Gauss-Seidel version [Breheny and Huang (2011),  Mazumder et al. (2011)] and coordinate descent of the Jacobian version, i.e., the iterative thresholding method
[Blumensath and Davies (2008), She (2009)].
These algorithms generate a sequence at which the objective functions are nonincreasing, but the convergence of the sequence itself is generally unknown. Moreover, if the sequence generated from multi-stage convex relaxation  (starts from a Lasso solution) converges,  it converges to some stationary point which may enjoy certain oracle statistical properties [Zhang (2010b), Fan, Xue and Zou (2014)] with the cost of a Lasso solver per iteration.
Jiao, Jin and Lu (2013) proposed a globally convergent primal dual active set algorithm for a class of nonconvex regularized problems. Recently, there has been much effort to show that CCCP, LLA and the path following proximal-gradient method can track the local minimizers with the desired statistical properties [Wang, Kim and Li (2013), Fan, Xue and Zou (2014), Wang, Liu and Zhang (2014) and Loh and Wainwright (2015)].

Another line of research includes greedy methods such as the orthogonal match pursuit (OMP)
[Mallat and Zhang (1993)] for solving (\ref{l0}) approximately.
The main idea is to iteratively select one variable with the strongest correlation with the current residual at a time.
Roughly speaking, the performance of OMP can be guaranteed if the small submatrices of $\bX$ are well conditioned   
like orthogonal matrices [Tropp (2004), Donoho,  Elad and Temlyakov (2006), Cai and Wang (2011), Zhang (2011a)].
Fan and Lv (2008) proposed a marginal correlation learning method called sure independence screening (SIS), see also Huang, Horowitz and Ma (2008) with an equivalent formulation
that uses penalized univariate regression for screening.
Fan and Lv (2008) recommended an iterative SIS to improve the finite-sample performance.
As they discussed the iterative SIS also uses the core idea of OMP but it
can select more  features at each iteration.
There are several more recently developed greedy methods aimed at selecting several variables a time or removing variables adaptively,
such as 
hard thresholding gradient descent (GraDes) [Garg and  Khandekar (2009)], stagewise OMP (StOMP) [Donoho et al. (2012)],
 adaptive forward-backward selection (FoBa) [Zhang (2011b)].

\subsection{Contributions}
SDAR is a new approach for fitting sparse, high-dimensional regression models.
Compared with the penalized methods, SDAR does not aim to minimize any regularized criterion, instead, it constructively generates a sequence of solutions $\{\beta^k, k \ge 1\}$ to the KKT equations of the $\ell_0$ penalized criterion.
SDAR can be viewed as a primal-dual active set method for solving the $\ell_{0}$ regularized least squares problem with a changing regularization parameter $\la$ in each iteration (this will be explained in detail in Section 2). However, the sequence generated by SDAR is not a minimizing sequence of the $\ell_{0}$ regularized least squares criterion.
Compared with the greedy methods, the features selected by SDAR are based on the sum of the primal (current approximation $\bbk$) and the dual information (current correlation  $\dk = \bX'(y-\bX\bbk)/n$), while greedy methods
only use  dual information. The differences between
SDAR and several greedy methods will be explained in more detail in Subsection
\ref{Relatedwork}.

We show that SDAR achieves sharp estimation error bounds in finite iterations.
Specifically, we show that: (a) under a sparse Rieze condition on $X$ and a sparsity assumption on $\beta^{*}$, $\|\beta^k-\bbp\|_{2}$ achieves the minimax error bound up to a constant factor with high probability in
$O(\sqrt{J}\log(R))$ steps; (b) under a mutual coherence condition on $X$ and a sparsity assumption on $\beta^*$, the $\|\beta^k-\bbp\|_{\infty}$ achieves the optimal error bound $O(\sigma \sqrt{\log(p)/n})$  in $O(\log(R))$ steps,
where $J$ is the number of important predictors, 
$R$ is the relative magnitude of the nonzero target coefficients
(the exact definitions of $J$ and $R$ are given in Section 3);
(c) under the conditions in (a) and (b),
with high probability,  $\beta^k$ coincides with the oracle least squares estimator in
$O(\sqrt{J}\log(R))$ and $O(\log(R))$ iterations, respectively,
if $J$ is available and the minimum magnitude of the nonzero elements of $\beta^{*}$ is of the order $O(\sigma \sqrt{2\log(p)/n})$, which is the optimal magnitude of detectable signal.

An interesting aspect of the result in (b) is that the number of iterations for SDAR to achieve the optimal error bound is $O(\log(R))$, which does not depend on the underlying sparsity level. This is an appealing feature for the problems with a large triple $(n, p, J)$.
We also analyze the computational cost of SDAR and show that it is $O(np)$ per iteration, comparable to the existing penalized and greedy methods.


In summary, the main contributions of this paper are as follows.

\begin{itemize}

\item We proposed a new approach to fitting sparse, high-dimensional regression models.
Unlike the existing penalized methods that
approximate the $\ell_0$ penalty using the $\ell_1$ or its concave modifications,
the proposed approach seeks to directly approximate the solutions to the
$\ell_0$ penalized problem.

\item We show that the sequence of solutions $\{\beta^k, k \ge 1\}$
generated by the SDAR achieves sharp error bounds. An interesting aspect
of our results is that these bounds can be achieved in $O(\sqrt{J}\log(R))$ or $O(\log(R))$  iterations.



\item
We also consider an adaptive version of SDAR, named ASDAR, by tuning the size of the fitted model based on a data driven procedure such as the BIC.
Our simulation studies
demonstrate that SDAR/ ASDAR outperforms the Lasso, MCP and several greedy methods in terms of accuracy and efficiency. 
\end{itemize}

\subsection{Notation}
Let $\|\beta\|_q = (\sum_{i=1}^{p}|\beta_{i}|^q)^{1/q}, q\in [1,\infty],$  be
the $q$-norm of a column vector
$\beta= (\beta_{1},\ldots,\beta_{p})^{\prime} \in \mathbb{R}^{p}$.
We denote the  number of nonzero elements of $\beta$ by $\|\beta\|_{0}$.
We denote 
the operator norm of $X$ induced by the vector 2-norm by $\norm X$.
We use \textbf{E} to denote the identity matrix.
$\textbf{0}$ denotes a column vector in $\mathbb{R}^{p}$ or a matrix with elements all 0.
Let $S =\{1,2,...,p\}$.  For any $ A,B\subseteq S$ with length $|A|,|B|$, we
denote $\beta_{A}=(\beta_{i}, i\in A)\in \mathbb{R}^{|A|}$,  $X_{A}=(\bX_{i}, i\in A)\in \mathbb{R}^{n\times|A|}$.  And $\bX_{AB}\in \mathbb{R}^{|A|\times
|B|}$  denotes a submatrix of $\bX$ whose rows and columns
are listed in $A$ and $B$, respectively.
We define $\beta |_{A}\in \mathbb{R}^{p}$ with its $i$-th element
$(\beta |_{A})_i = \beta_i \textbf{1}(i\in A)$,  where $\textbf{1}(\cdot)$ is the indicator function. We denote the support of $\beta$ by $\textrm{supp}(\beta)$.
Define $ A^{*} = \textrm{supp}(\bbp)$  and  $K = \normz{\bbp}$.
We use $\|\beta\|_{k,\infty}$ and $\abs{\beta}_{\textrm{min}}$ to denote the $k$-th largest elements (in absolute value) of $\beta$ and the  minimum absolute value of $\beta$, respectively.

\subsection{Organization}
In Section 2 we
develop the SDAR algorithm based on the necessary conditions for the $\ell_0$ penalized solutions. In Section 3 we establish the nonasymptotic error bounds of the SDAR solutions.
In Section 4 we describe a data driven adaptive SDAR. In Section 5 we analyze the computational complexity of SDAR and ASDAR and  discuss the relationship of SDAR with several greedy methods and screening method.
In Section 6 we conduct simulation studies to demonstrate the 
performance of SDAR/ ASDAR by comparing  it with Lasso, MCP and several greedy methods.
We conclude in Section 7 with some final remarks. The proofs are given in the Appendix.

\section{Derivation of SDAR}
In this section we describe the SDAR algorithm.
Consider the Lagrangian form of the $\ell_{0}$ regularized minimization problem  (\ref{l0}),
\begin{equation}\label{eq1}
\min_{\bb \in \bRp } \frac{1}{2n}\normt{\bX\bb-y}^2 + \la\normz{\bb}.
\end{equation}

\begin{lem}\label{lem1}
Let  $\beta^\diamond$ be a minimizer of \eqref{eq1}.  Then $\beta^\diamond$ satisfies:
\begin{equation}\label{eq2}
\begin{cases}
d^\diamond = \bX'(y-\bX\beta^\diamond)/n, \\
\beta^\diamond = H_{\la}(\beta^\diamond + d^\diamond),
\end{cases}
\end{equation}
where $H_{\la}(\cdot)$ is the hard thresholding operator defined by
\begin{equation}\label{hardth}
(H_{\la}(\beta))_{i}=
\begin{cases}
\textrm{0}, & \text{if $|\beta_{i}|<\sqrt{2\la}$},\\
\text{$\beta_{i}$}, & \text{if $|\beta_{i}|\geq \sqrt{2\la}$}.
\end{cases}
\end{equation}
Conversely, if $\beta^\diamond$ and $d^{\diamond}$ satisfy \eqref{eq2}, then $\beta^\diamond$ is a local minimizer of \eqref{eq1}.
\end{lem}

\begin{remark}
Lemma \ref{lem1} gives the KKT condition of the $\ell_0$ regularized minimization problem  (\ref{l0}), similar results for SCAD 
MCP 
capped-$\ell_1$ 
regularized least squares models can be derived by replacing the hard thresholding operator in \eqref{eq2} with their corresponding thresholding operators, see Jiao, Jin and Lu (2013) for details.
\end{remark}

Let $\As = \textrm{supp}(\bbs)$ and $ \Is=(\As)^c$. Suppose that the rank of $\bX_{\As}$ is $|\As|$. From the definition of $H_\la(\cdot)$ and \eqref{eq2} it follows that
\begin{equation*}
\As = \dkhB{ i \in S \big| \abs{\bbs_i + \ds_i} \geq \sqrt{2\la}},\quad \Is=\dkhB{ i\in S \big| \abs{\bbs_i + \ds_i} < \sqrt{2\la}}, \end{equation*}
and
\begin{equation*}
\left\{
\begin{aligned}
\bbs_{\Is} &= \textbf{0}, \\
\ds_{\As} &=\textbf{0}, \\
\bbs_{\As} &=(\bX_{\As}'\bX_{\As})^{-1}\bX_{\As}'y,\\
\ds_{\Is} &= \bX'_{\Is}(y-\bX_{\As}\bbs_{\As})/n.
\end{aligned}
\right.
\end{equation*}
We solve this system of equations iteratively.
Let $\{\bbk, \dk\}$ be the solution at the $k$th iteration.
We approximate $\{\As, \Is\}$ by
\begin{equation}
\label{DefAk}
\Ak = \dkhB{ i\in S \big| |\bbk_i + \dk_i| \geq \sqrt{2\la}},\quad \Ik=(\Ak)^c.
\end{equation}
Then we can obtain a new approximation pair $(\bbkk, \dkk)$ by
\begin{equation}
\label{Stepk}
\left\{
\begin{aligned}
\bbkk_{\Ik} &= \textbf{0}, \\
\dkk_{\Ak} &= \textbf{0}, \\
\bbkk_{\Ak} &= (\bX_{\Ak}'\bX_{\Ak})^{-1}\bX_{\Ak}'y, \\
\dkk_{\Ik} &= \bX'_{\Ik}(y-\bX_{\Ak}\bbkk_{\Ak})/n.
\end{aligned}
\right.
\end{equation}
Now suppose we want the support of the solutions to have size
$T\ge 1$. We can choose \begin{equation}
\label{lam1}
\sqrt{2\la} \de \|{\bbk + \dk}\|_{T,\infty}
\end{equation}
in (\ref{DefAk}). With this choice of $\la$, we have  $|\Ak|=T, k \ge 1.$
Then with an initial $\beta^0$ and using  (\ref{DefAk}) and (\ref{Stepk}) with the $\la$ in (\ref{lam1}), we obtain a sequence of solutions $\{\beta^k, k \ge 1\}$.

There are two key aspects of SDAR. In (\ref{DefAk}) we detect the support
of the solution based on the sum of the primal ($\bbk$) and dual ($\dk$) approximations and, in (\ref{Stepk}) we calculate the nonzero solution on the detected support. Therefore, SDAR can be considered an iterative method for solving the KKT equations (\ref{eq2}) with an important modification: a different $\la$ value given in (\ref{lam1}) in each step of the iteration is used.
Thus we can also view SDAR as an adaptive thresholding  and  least-squares fitting  procedure
that uses both the primal and dual information.
We summarize SDAR in Algorithm \ref{alg1}.

\begin{remark}
If  $\Akk = \Ak$ for some $k$ we stop SDAR since the sequences generated by SDAR will  not change. Under certain conditions,  we will show  that $\Akk = \Ak = \textrm{supp}(\beta^*)$ if $k$ is large enough, i.e., the stop condition in SDAR will be active and the output is  the oracle estimator when it stops.
 \end{remark}

As an example, Figure \ref{fig1} shows the solution path of SDAR with
$T=1, 2, \ldots, 5K$,   
the MCP  and the Lasso paths
on $5K$ different $\lambda$ values for a data set generated from a model with
$(n=50, p=100, K=5, \sigma=0.3, \rho = 0.5, R=10)$, which will be described in Section 6.
The Lasso path is computed using LARS [Efron et al. (2004)].
Note that the SDAR path is a function of the fitted model size $T=1, \ldots, L$, where
$L$ is the size of the largest fitted model. In comparison, the paths of
MCP and Lasso are functions of the penalty parameter $\lambda$ in prespecified interval.
In this example, SDAR selects the first $T$ largest components of $\bbp$ correctly when
$T \leq K$.

\begin{algorithm}[H]
\caption{Support detection and root finding (SDAR)}
\label{alg1}
\begin{algorithmic}[1]
\REQUIRE
$\beta^0$, $d^0=\bX'(y-\bX\beta^0)/n$, $T$; set $k=0$.
\FOR {$k=0,1,2,\cdots$}
\STATE $\Ak = \{ i \in S \big| |\beta^k_{i} + d^k_{i}| \geq \|\beta^k+d^k\|_{T,\infty}\},\quad \Ik=(\Ak)^c$
\STATE $\bbkk_{\Ik} = \textbf{0}.$
\STATE $\dkk_{\Ak} = \textbf{0}.$
\STATE $\bbkk_{\Ak} = (\bX_{\Ak}'\bX_{\Ak})^{-1}\bX_{\Ak}'y.$
\STATE $\dkk_{\Ik} = \bX'_{\Ik}(y-\bX_{\Ak}\bbkk_{\Ak})/n.$
 \IF {$\Akk = \Ak$,}
 \STATE Stop and denote the last iteration by $\beta_{\hat{A}}, \beta_{\hat{I}}, d_{\hat{A}}, d_{\hat{I}}$;
 \ELSE
 \STATE $k = k+1$.
 \ENDIF
\ENDFOR
\ENSURE $\hat{\beta}= (\beta_{\hat{A}}^{\prime},\beta_{\hat{I}}^{\prime})^{\prime}$
as the estimates of $\beta^*$.
\end{algorithmic}
\end{algorithm}

\begin{figure}[H]
\centering
\begin{tabular}{ccc}
\includegraphics[trim = 0.5cm 0cm 0.5cm 0cm, clip=true,width=4cm]{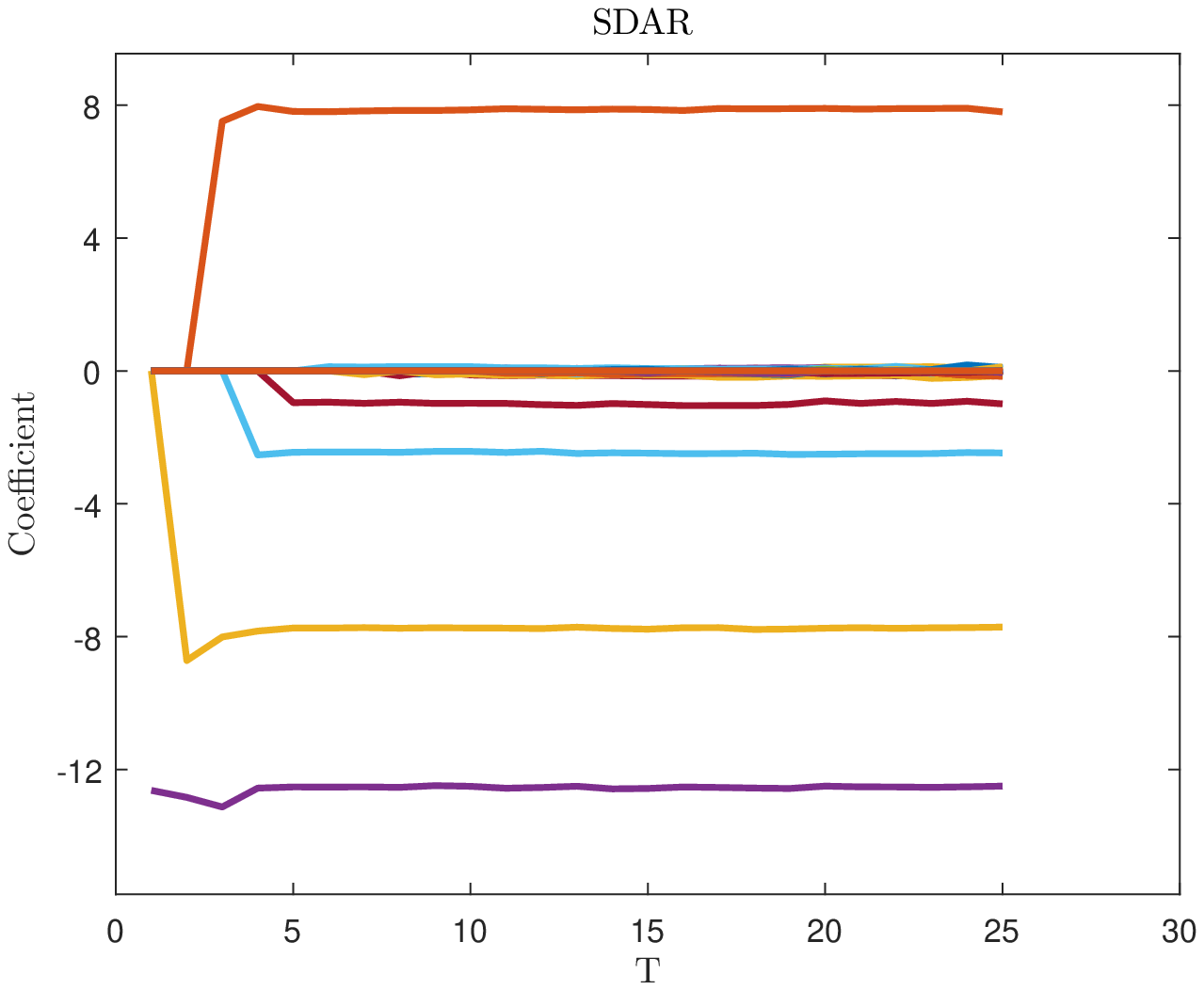}&
\includegraphics[trim = 0.5cm 0cm 0.5cm 0cm, clip=true,width=4cm]{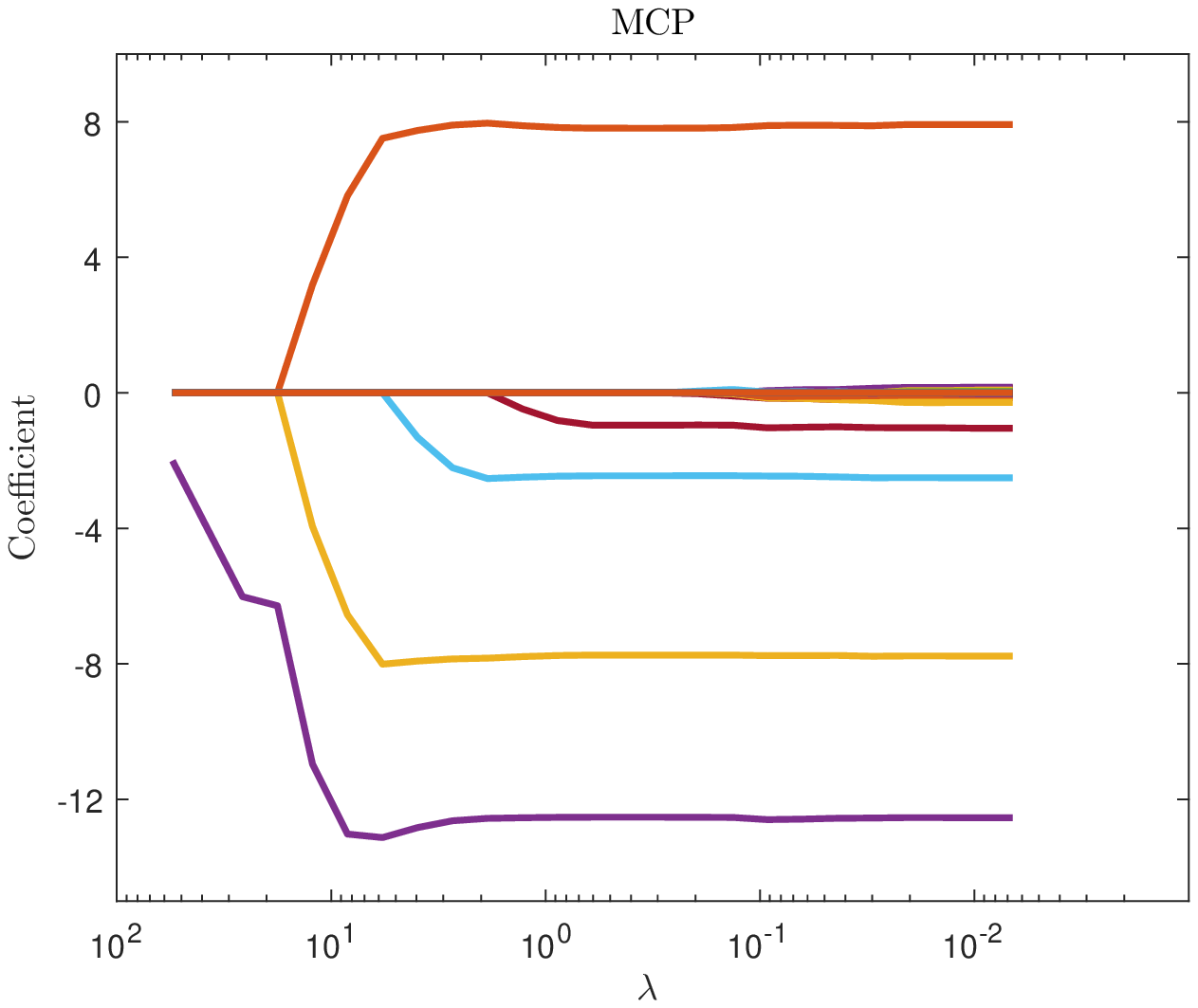}&
\includegraphics[trim = 0.5cm 0cm 0.5cm 0cm, clip=true,width=4cm]{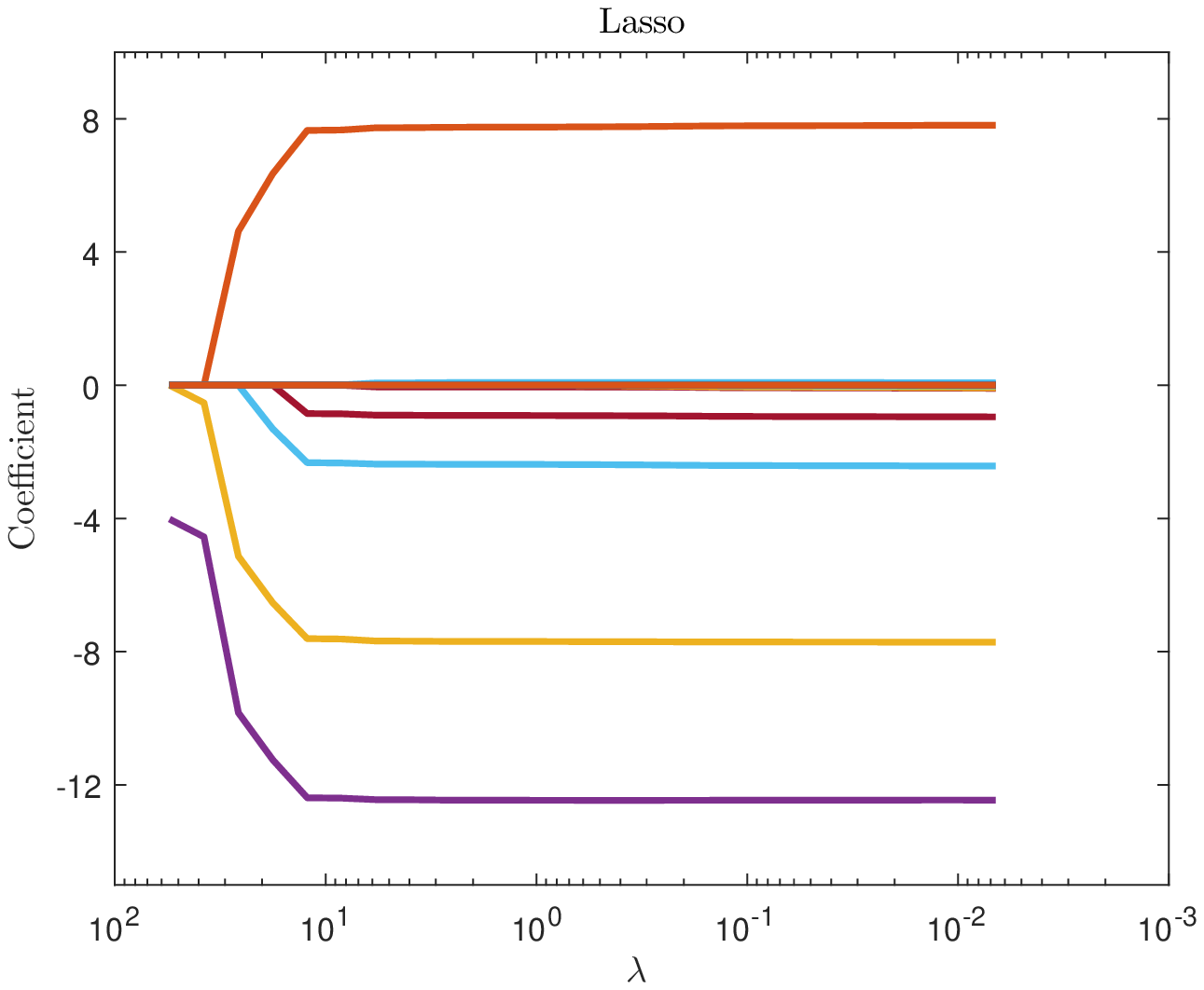}\\
\end{tabular}
\caption{The solution paths of SDAR,  MCP and Lasso.}\label{fig1}
\end{figure}


\section{Nonasymptotic error bounds}
In this section we present the nonasymptotic $\ell_2$ and $\ell_{\infty}$ error bounds
for the solution sequence generated by
SDAR as given in Algorithm \ref{alg1}.


We say that $\bX$ satisfies the SRC [Zhang and Huang (2008), Zhang (2010a)] with order $s$ and spectrum bounds $\{\cm(s), \cp(s)\}$ if
\[
0 < \cm(s)\leq\frac{\normt{\bX_A u}^2}{n\normt{u}^2}\leq\cp(s) < \infty,
\forall \, {0}\neq u\in \mathbb{R}^{\abs{A}}
\mbox{ with } A\subset S \mbox{ and } \abs{A}\leq s.
\]
We denote this condition by $\bX \sim \textrm{SRC}\{s,\cm(s), \cp(s)\}$.
The SRC gives the range of the spectrum of the diagonal sub-matrices of the Gram matrix $G=\bX^{\prime}\bX/n$.
The spectrum of the off diagonal sub-matrices of $G$ can be bounded by the sparse orthogonality constant $\theta_{a,b}$ defined as the smallest number such that
\[
 \theta_{a,b} \ge \frac{\normt{\bX_A^{\prime} \bX_B u}}{n\normt{u}},
\forall \, \textrm{0}\neq u\in \mathbb{R}^{\abs{B}} \mbox{ with }A,B\subset S,  \abs{A}\leq a, \abs{B}\leq b, \mbox{ and } A\cap B = \emptyset.
\]
Another useful quantity is the mutual coherence $\mu$ defined as $\mu = \max_{i\neq j}\abs{G_{i,j}}$, which characterizes the minimum angle between different columns of $\bX/\sqrt{n}$.
Some useful properties of these quantities are summarized in 
Lemma \ref{lem2} in the Appendix.

In addition to the regularity conditions on the design matrix, another key condition is the sparsity of the regression parameter $\bbp$. The usual sparsity condition is to assume that the regression parameter $\bbp_i$ is either nonzero or zero and that the number of nonzero coefficients is relatively small. This strict sparsity condition is not realistic in many problems.
Here we allow that $\bbp$ may not be strictly sparse but most of its elements are small.
Let $A^{*}_J=\{i\in S: |\bbp_i| \geq \|\bbp\|_{J,\infty}\}$ be the set
of the indices  of the first $J$ largest components of  $\bbp$.
Typically, we have $J \ll n$.
Let
\begin{equation}
\label{Rbar}
R=\frac{\bar{M}}{\bar{m}},
\end{equation}
where $\bar{m} = \min\{|\bbp_{i}|, i\in A^{*}_J\}$ and $ \bar{M} = \max\{|\bbp_{i}|, i\in A^{*}_J\}$.
Since $\bbp = \bbp|_{A^{*}_J} +  \bbp|_{(A^{*}_{J})^c}$,  we can transform the non-exactly sparse model (\ref{model}) to the following  exactly sparse model by including the small components of $\beta^{*}$ in the noise,
\begin{equation}\label{model*}
y = \bX\tb + \te,
\end{equation}
where
\begin{equation}\label{teta}
\tb = \bbp|_{A^{*}_J} \ \mbox{ and } \
\te = \bX\bbp|_{(A^{*}_{J})^c}+ \eta.
\end{equation}
Let $R_{J} = \normt{\bbp|_{(A^{*}_{J})^c}}+\normo{\bbp|_{(A^{*}_{J})^c}}/\sqrt{J}$, which is a measure of the magnitude of the small components of $\bbp$ outside $A^{*}_J$.
Of course, $R_{J}=0$ if $\bbp$ is exactly $K$-sparse with $K\leq J$. Without loss of generality, we let $J = K$, $m = \bar{m}$ and $M = \bar{M}$  for simplicity   
if $\bbp$ is exactly $K$-sparse.

Let $\beta^{J,o}$ be the oracle estimator defined as
$ \beta^{J,o}=\arg\min_{\beta}\{\frac{1}{2n}\|y-X\beta\|_2^2, \beta_j=0, j \not \in A^{*}_J\}$,
that is,  
$\beta^{J,o}_{A^{*}_{J}} = \bX_{A^{*}_{J}}^{\dag}y$
and
$\beta^{J,o}_{(A^{*}_{J})^c} = \textbf{0}$, where $\bX_{A^{*}_{J}}^{\dag}$ is the generalized inverse of $\bX_{A^{*}_{J}}$ and equals to  $(\bX_{A^{*}_{J}}^{\prime}\bX_{A^{*}_{J}})^{-1}\bX_{A^{*}_{J}}^{\prime}$ if $\bX_{A^{*}_{J}}$ is of full column rank.
So $\beta^{J,o}$ is obtained by keeping the predictors corresponding to the
$J$ largest components of $\bbp$ in the model and dropping the other predictors. Obviously,  $\beta^{J,o} = \beta^{o}$ if $\bbp$ is exactly $K$-sparse, where $\beta^{o}_{A^{*}} = \bX_{A^{*}}^{\dag}y, \beta^{o}_{(A^{*})^c} = \textbf{0}$.

\subsection{$\ell_{2}$ error bounds}
Let $1\le T \le p$ be a given integer used in Algorithm \ref{alg1}.
We require the following basic assumptions on the design matrix $\bX$ and
the error vector $\eta$.

(A1) The input integer $T$ used in Algorithm \ref{alg1} satisfies  $ T\geq J$.

(A2) $\bX \sim \text{SRC}\{2T,\cm(2T), \cp(2T)\}$.




(A3) The random errors $\eta_1, \ldots, \eta_n$ are independent and identically
distributed with mean zero and sub-Gaussian tails, that is, there exists a $\sigma\geq 0$ such that $E[\exp(t\eta_i)]\leq \exp(\sigma^2t^2/2)$ for $t \in \mathbb{R}^{1}$, $i = 1,\ldots, n.$

Let
\[
\gamma= \frac{2\theta_{T,T}+(1+\sqrt{2})\theta_{T,T}^2}{\cm(T)^2}+ \frac{(1+\sqrt{2})\theta_{T,T}}{\cm(T)}.
\]
Define
$\htt = \max_{A\subseteq S:|A|\le T}\normt{\bX_{A}^{\prime}\te}/n$,
where $\te$ is defined in (\ref{teta}).

\begin{thm}\label{thmt}
Let $1\le T \le p$ be a given integer used in Algorithm \ref{alg1}.
Suppose $\gamma < 1$.
\begin{itemize}
\item[(i)] Assume   \textrm{(A1) and  (A2)} hold. We have
\begin{align}
\normt{\tb|_{A_{J}^* \backslash \Akk}} & \leq \gamma^{k+1}\normt{\tb} + \frac{\gamma }{(1-\gamma)\theta_{T,T}}\htt,\label{thmt-1} \\
\normt{\bbkk-\tb}& \leq b_1\gamma^k\normt{\tb}+ b_2\htt, \label{thmt-2}
\end{align}
where
\begin{equation}
\label{b12}
b_1 = 1+\frac{\theta_{T,T}}{\cm(T)} \ \mbox{ and } \
b_2=\frac{\gamma }{(1-\gamma)\theta_{T,T}}b_1 + \frac{1}{\cm(T)}.
\end{equation}

\item[(ii)] Assume  (A1)-(A3) hold. 
Then  for any $\alpha\in (0,1/2)$, with probability at least $1 - 2\alpha$,
\begin{align}
\normt{\tb|_{A_{J}^* \backslash \Akk}} &\leq \gamma^{k+1}\normt{\tb} + \frac{\gamma }{(1-\gamma)\theta_{T,T}}\vps_1,\label{thmt-3}\\
\normt{\bbkk-\tb}&\leq b_1\gamma^k\normt{\tb}+ b_2\vps_1, \label{thmt-4}
\end{align}
where
$
 \vps_1 = \cp(J)R_{J} + \s\sqrt{T}\sqrt{2\log(p/\alpha)/n}.
$
\end{itemize}

\end{thm}

\begin{remark}
Assumption (A1) is necessary for SDAR to select at least $J$ nonzero features.
The sparse Riesz condition in (A2) has been used in the analysis of the Lasso and MCP
[Zhang and Huang (2008), Zhang (2010a)].
Let $c(T) = (1-\cm(2T)) \vee (\cp(2T)-1)$, which is closely related to the
the RIP (restricted isometry property) constant $\delta_{2T}$ for $\bX$ [Cand\`{e}s and Tao (2005)]. By \eqref{lem2-4} in the Appendix, it can be verified  that a sufficient condition
for $\gamma < 1$ is $c(T)\leq0.1599$, i.e.,  $\cp(2T) \leq 1.1599$, $\cm(2T) \geq 0.8401$.
%
%
The sub-Gaussian condition (A3) is often assumed in the literature and slightly weaker than the standard normality assumption.
\end{remark}

\begin{remark}
Several greedy algorithms have  also been studied under the assumptions related to the sparse Riesz condition. For example,  Zhang (2011b) studied OMP under the  condition $c_{+}(T)/c_{-}(31T)\leq 2$.  Zhang (2011a)  analyzed the forward-backward greedy algorithm (FoBa) under the condition $8(T+1)\leq (s-2)T c_{-}^2(sT)$, where $s>0$ is a properly chosen parameter.
 GraDes [Garg and  Khandekar (2009)]
has been analyzed  under the RIP condition $\delta_{2T}\leq 1/3$.
These  conditions and (A2) are related but do not imply each other.
The order of $\ell_2$-norm  estimation error of SDAR is at least as good as that of the above mentioned greedy methods since it achieves the minimax error bound, see, Remark \ref{mimimax} below.
A high level comparison of SDAR with the greedy algorithms will be given in Section 5.2.
\end{remark}

\begin{cor}\label{cort}
\begin{itemize}
\item[(i)] Suppose  \textrm{(A1) and  (A2)} hold. Then
\begin{equation}\label{corterrb}
\normt{\bbk - \tb} \leq c\htt \quad\textrm{if}\quad k\geq \log_{\frac{1}{\gamma}}\frac{\sqrt{J}\bar{M}}{\htt}.
\end{equation}
where $c =b_1+b_2$ with $b_1$ and $b_2$ defined in (\ref{b12}).

Further assume  $\bar{m} \geq \frac{\gamma \htt}{(1-\gamma)\theta_{T,T}\xi}$ for some $0<\xi< 1$, then,
\begin{equation}\label{cortacb}
\Ak \supseteq A^{*}_{J} \quad\textrm{if}\quad  k\geq\log_{\frac{1}{\gamma}}\frac{\sqrt{J}R}{1-\xi}.
\end{equation}
\item[(ii)] Suppose (A1)-(A3) hold. Then, for any $\alpha\in (0,1/2)$,
 with probability at least $1 - 2\alpha$, we have
\begin{equation}\label{corterrg}
\normt{\bbk - \tb} \leq  c \vps_1 \quad\textrm{if}\quad k\geq \log_{\frac{1}{\gamma}}\frac{\sqrt{J}\bar{M}}{\vps_1}.
\end{equation}
Further assume   $ \bar{m} \geq \frac{\vps_1\gamma}{(1-\gamma)\theta_{T,T}\xi}$ for some $0 <\xi< 1$, then, with probability at least $1 - 2\alpha$
\begin{equation}\label{cortacg}
\Ak \supseteq A^{*}_{J} \quad\textrm{if}\quad  k\geq\log_{\frac{1}{\gamma}}\frac{\sqrt{J}R}{1-\xi}.
\end{equation}
\item[(iii)]
Suppose $\beta^*$ is exactly  $K$-sparse.
Let  $T=K$ in SDAR. Suppose (A1)-(A3) hold and  $ m \geq \frac{\gamma}{(1-\gamma)\theta_{T,T}\xi}\s\sqrt{K}\sqrt{2\log(p/\alpha)/n}$ for some $0 <\xi< 1$,
we have with probability at least $1-2\alpha$, $A^k= A^{k+1} = A^*$ if $k\geq\log_{\frac{1}{\gamma}}\frac{\sqrt{K}R}{1-\xi}$, i.e., using   at most $O(\log \sqrt{K}R)$ iterations, SDAR stops and the output is the oracle estimator $\beta^{o}$.
\end{itemize}
\end{cor}

\begin{remark}\label{mimimax}
Suppose   $\beta^*$ is exactly $K$-sparse.
In the event $\normt{\eta}\leq\vps$, part (i) of Corollary \ref{cort} implies $\normt{\bbk-\beta^*}=O(\vps/\sqrt{n})$  if $k$ is sufficiently large.
Under certain conditions on the RIP constant of $\bX$, Cand\`{e}s, Romberg and Tao (2006) showed that $\|\hat{\beta}-\beta^*\|_2 = O(\vps)$,
where $\hat{\beta}$ solves
\begin{equation}\label{CRT}
 \min_{\beta\in \mathbb{R}^{p}} \|\beta\|_{1} \mbox{ subject to }  \normt{\bX\beta-y}\leq\vps.
\end{equation}
So the result here is similar to that of Cand\`{e}s, Romberg and Tao (2006) (There is a factor $1/\sqrt{n}$ in our result since we assume the columns of $\bX$ are $\sqrt{n}$-length   normalized while  they assumed the columns of $\bX$  are unit-length normalized).
However, it is a nontrivial task to solve \eqref{CRT} in high-dimensional settings. In comparison, SDAR only involves simple computational steps.
\end{remark}
\begin{remark}
Let  $\beta^*$ be exactly $K$-sparse.
Part (ii) of   Corollary \ref{cort} implies that SDAR achieves the minimax
error bound [Raskutti, Wainwright   and  Yu  (2011)], that is,
$$\normt{\bbk - \beta^*} \leq c \s\sqrt{T}\sqrt{2\log(p/\alpha)/n}$$ with high probability if $k\geq \log_{\frac{1}{\gamma}}\frac{\sqrt{K}M}{\s\sqrt{T}\sqrt{2\log(p/\alpha)/n}}$.
\end{remark}


\subsection{$\ell_{\infty}$ error bounds}
We now consider the $\ell_{\infty}$ error bounds of SDAR.
We replace condition (A2) by

\smallskip
(A2*) The mutual coherence $\mu$ of $\bX$ satisfies   $T\mu \leq 1/4$.



 Let
$\gamma_{\mu} = \frac{(1+2T\mu)T\mu}{1-(T-1)\mu}+2T\mu$,  $c_{\mu}=\frac{16}{3(1-\gamma_{\mu})}+\frac{5}{3}$ and
 $\hit = \max_{A \subseteq S:|A|\le T}\normi{\bX_{A}^{\prime}\te}/n$, where $\te$ is defined in (\ref{teta}).

\begin{thm}\label{thmi}
Let $1\le T \le p$ be a given integer used in Algorithm \ref{alg1}. 

\begin{itemize}
\item[(i)] Assume   \textrm{(A1) and  (A2*)} hold. We have
\begin{align}
\normi{\tb|_{{A^{*}_J} \backslash \Akk}} &< \gamma_{\mu}^{k+1}\normi{\tb} + \frac{4}{1-\gamma_{\mu}}\hit,\label{thmi-1}\\
\normi{\bbkk-\tb} &< \frac{4}{3}\gamma_{\mu}^k\normi{\tb}+\frac{4}{3}
(\frac{4}{1-\gamma_{\mu}}+1)\hit, \label{thmi-2}
\end{align}


\item[(ii)] Assume (A1), (A2*) and (A3) hold.
For any $\alpha\in (0,1/2)$,  with probability at least $1 - 2\alpha$,
\begin{align}
\normi{\tb|_{{A^{*}_J} \backslash \Akk}} &< \gamma_{\mu}^{k+1}\normi{\tb} + \frac{4}{1-\gamma_{\mu}}\vps_2,\label{thmi-3}\\
\normi{\bbkk-\tb} &< \frac{4}{3}\gamma_{\mu}^k\normi{\tb}+\frac{4}{3}
(\frac{4}{1-\gamma_{\mu}}+1)\vps_2, \label{thmi-4}
\end{align}
where
$
\vps_2= (1+(T-1)\mu)R_{J} +  \s\sqrt{2\log(p/\alpha)/n}.
$
\end{itemize}
\end{thm}



\begin{cor}\label{cori}
\begin{itemize}
\item[(i)] Suppose (A1) and (A2*) hold. Then
\begin{equation}\label{corierrb}
\normi{\bbk - \tb} \leq  c_{\mu} \hit \quad\textrm{if}\quad k\geq \log_{\frac{1}{\gamma_{\mu}}}\frac{4\bar{M}}{\hit}.
\end{equation}

Further assume  $\bar{m} \geq \frac{4\hit}{(1-\gamma_{\mu})\xi}  $ with $\xi< 1$, then,
\begin{equation}\label{coriacb}
\Ak \supseteq A^{*}_{J} \quad\textrm{if}\quad  k\geq\log_{\frac{1}{\gamma_{\mu}}} \frac{R}{1-\xi}.
\end{equation}

\item[(ii)] Suppose (A1), (A2*) and (A3) hold.  Then  for any $\alpha\in (0,1/2)$, with probability at least $1 - 2\alpha$,
\begin{equation}\label{corierrg}
\normi{\bbk - \tb} \leq  c_{\mu} \vps_2 \quad\textrm{if}\quad k\geq \log_{\frac{1}{\gamma_{\mu}}}\frac{4\bar{M}}{\vps_2}.
\end{equation}

Further assume  $\bar{m} \geq \frac{4\vps_2}{\xi (1-\gamma_{\mu})}$ for some $0<\xi<1$, then,
\begin{equation}\label{coriacg}
\Ak \supseteq A^{*}_{J} \quad\textrm{if}\quad  k\geq\log_{\frac{1}{\gamma_{\mu}}} \frac{R}{1-\xi}.
\end{equation}
\item[(iii)]
Suppose $\beta^*$ is exactly  $K$-sparse.
Let  $T=K$ in SDAR. Suppose (A1), (A2*), (A3) hold and $m\geq \frac{4}{\xi (1-\gamma_{\mu})}\s\sqrt{2\log(p/\alpha)/n}$ for some $0<\xi<1$,
we have with probability at least $1-2\alpha$, $A^k= A^{k+1} = A^*$ if $k\geq\log_{\frac{1}{\gamma_{\mu}}} \frac{R}{1-\xi}$, i.e., with
at most $O(\log R)$ iterations, SDAR stops and the output is the oracle
least squares estimator $\beta^{o}$.
\end{itemize}
\end{cor}

\begin{remark}

Theorem \ref{thmt} and Corollary \ref{cort} can be derived
from Theorem \ref{thmi} and Corollary \ref{cori}, respectively, by using the relationship between the $\ell_{\infty}$ norm and the $\ell_2$ norm. Here we present them separately because (A2) is weaker than (A2*). The stronger assumption (A2*) brings us some new insights into the
 SDAR, i.e., the sharp $\ell_{\infty}$ error bound, based on which we can show that the worst case iteration complexity of SDAR does not depend on the underlying sparsity level, see, Corollary \ref{cori}.
\end{remark}

\begin{remark}
The mutual coherence condition  $s\mu\leq 1$ with $s\geq 2K-1$ is
used in the study of OMP and Lasso in the case that $\beta^*$ is exactly $K$-sparse.
In the noiseless case with $\eta=0$, Tropp (2004) and Donoho and  Tsaig (2008) showed that under the condition $(2K-1)\mu<1$,  OMP can recover $\beta^{*}$ exactly in $K$ steps. In the noisy case with $\normt{\eta} \le \vps$, Donoho,  Elad and Temlyakov (2006) proved that OMP can recover the true support if $(2K-1)\mu\leq 1- \frac{2\vps}{m}$. Cai and Wang (2011) gave a sharp analysis of OMP under the condition  $(2K-1)\mu<1$.
The mutual coherence condition $T\mu \le 1/4$  in (A2*) is a little stronger than those used in the analysis of the OMP. However, under (A2*) we obtain a  sharp  $\ell_{\infty}$ error bound, which is not available for OMP in the literature.
Furthermore,  Corollary \ref{cori} implies that the number of iterations of SDAR does not depend on the sparsity level, which is a surprising result and does not  appear in the
literature on greedy methods,  see Remark \ref{ActiveSet2} below.
Lounici (2008) and  Zhang (2009) derived an $\ell_{\infty}$ estimation error bound
for the Lasso 
under the conditions  $K\mu<{1}/{7}$ and $K\mu\leq {1}/{4}$, respectively. However,
they need a nontrivial Lasso solver for computing an approximate solution while SDAR only involves simple computational steps.
\end{remark}

\begin{remark}
Suppose  $\beta^*$ is exactly $K$-sparse.
Part (ii) of  Corollary \ref{cori} implies that the sharp error bound
\begin{equation}\label{linferr}
\normi{\bbk - \beta^*} \leq c_{\mu} \s\sqrt{2\log(p/\alpha)/n}
\end{equation}
 can be achieved with high probability if
$k\geq \log_{\frac{1}{\gamma_{\mu}}} \frac{M}{\s\sqrt{2\log(p/\alpha)/n}}$.
\end{remark}


\begin{remark}\label{ActiveSet2}
Suppose   $\beta^*$ is exactly $K$-sparse.
Part (iii) of  Corollary \ref{cori} implies that
with high probability,  the oracle estimator can be recovered in no more than $O(\log R)$ steps if we set $T=K$ in SDAR and the  minimum magnitude of the nonzero elements of $\beta^{*}$ is $O(\sigma \sqrt{2\log(p)/n})$, which is the optimal magnitude of detectable signals.

It is interesting to notice  that the number of iterations in Corollary \ref{cori}  depends
on the relative magnitude $R$, but not the sparsity level $K$, see, Figure \ref{fig2} for the numerical results supporting this. This improves the  result in  part  (iii) of  Corollary \ref{cort}.  This  is  a surprising result since as far as we know the number of iterations for greedy methods to recover $A^{*}$ depends on $K$, see for example, Garg and  Khandekar (2009).
\end{remark}

Figure \ref{fig2} shows the average number of iterations of SDAR  with $T = K$
based on 100 independent replications on data set $(n= 500, p = 1000, K= 3:2:50, \sigma = 0.01,\rho = 0.1, R = 1)$ which  will be described in Section 6. We can see that as the sparsity level increases from $3$ to $50$ the  average number of iterations of SDAR ranges from $1$ to $3$, which is  bounded by $O(\log \frac{R}{1-\xi})$ with a suitably chosen $\xi$.

\section{Adaptive SDAR}
In practice, the sparsity level of the model is usually unknown,
we can use a data driven procedure to determine $T$, an upper bound of number of important variables  $J$,
used in SDAR (Algorithm \ref{alg1}).
The idea is to take $T$ as a tuning parameter, so $T$ plays the role similar to the penalty parameter $\la$ in a penalized method.
We can run SDAR from $T=1$ to a large $T = L$. For example,  we can take $L=O(n/\log(n))$ as suggested by Fan and Lv (2008), which is an upper bound of the largest possible model that can be consistently estimated with sample size $n$. By doing so we obtain a solution path $\{\hbeta(T): T=0, 1,\ldots, L\}$,
where $\hbeta(0)=\textrm{0}$, that is, $T=0$ corresponds to the null model.
Then we use a data driven criterion, such as HBIC [Wang, Kim and Li (2013)], to select a $T=\hT$ and use
$\hbeta(\hT)$ as the final estimate.
The overall computational complexity of this process  is
$O(Lnp \log(R))$, see Section 5 (we can also compute the path by increasing $T$ geometrically which may be more efficient, but here we are interested in the  complexity of the worst case).
We note that tuning $T$ is no more difficult than tuning a continuous
penalty parameter $\la$ in a penalized method. Indeed, here we can simply increase $T$ one by one from $T=0$ to $T=L$.
In comparison, in tuning the value of $\la$ based on a pathwise solution on an interval $[\la_{\min}, \la_{\max}]$, where $\la_{\max}$ corresponds to the null model and $\la_{\min}>0$ is a small value.
We need to determine the grid of $\la$ values on $[\la_{\min}, \la_{\max}]$
as well as $\la_{\min}$. Here $\la_{\min}$ corresponds to the largest model on the solution path. In the numerical implementation of the coordinate
descent algorithms for the Lasso [Friedman et al. (2007)], MCP and
SCAD [Breheny and Huang (2011)], $\la_{\min}=\alpha \la_{\max}$ for a small $\alpha$, for example, $\alpha=0.0001$.
Determining the value of $L$ is somewhat
similar to determining $\la_{\min}$. However, $L$ has the  meaning
of the model size, but the meaning of $\la_{\min}$ is less explicit.

\begin{figure}[H]
\centering
\includegraphics[scale=0.55]{{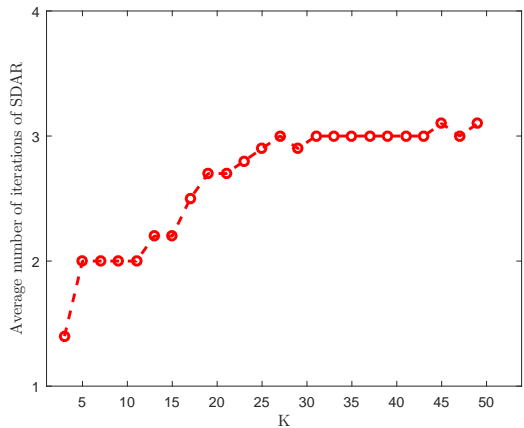}}
\caption{The average  number of iterations of SDAR  as $K$ increases.}\label{fig2}
\end{figure}

We also have the option to stop the iteration early according to other criterions. For example, we can run SDAR (Algorithm \ref{alg1}) by gradually increasing $T$  until  the change in the consecutive  solutions is smaller than a given value. Cand\`{e}s, Romberg and Tao (2006) proposed to recover $\bbp$ based on (\ref{CRT}) by finding the most sparse solution  whose residual sum of squares is smaller than a prespecified noise level $\vps$.  Inspired by this, we can also run SDAR by increasing $T$ gradually until the residual sum of squares is smaller than a prespecified value $\vps$.

We summarize these ideas in Algorithm 2 (Adaptive SDAR Algorithm) below.

\begin{algorithm}[!ht]
\caption{Adaptive SDAR (ASDAR)}
\label{alg2}
\begin{algorithmic}[1]
\REQUIRE
 Initial guess $\beta^0, d^0$, an integer $\tau$,
  an integer $L$,  and an early stopping criterion (optional).   Set $k =1$.
\FOR {$k =1,2,\cdots$}
\STATE  Run Algorithm \ref{alg1}  with  $ T = \tau k $ and with initial value $ (\beta^{k-1}, d^{k-1})$.  Denote the output by $(\beta^{k}, d^{k})$.
 \IF {the early stopping criterion is satisfied
  or $T>L$}
 \STATE stop
 \ELSE
 \STATE  $ k = k+1$.
 \ENDIF
\ENDFOR
\ENSURE $\hat{\beta}(\hat{T})$ as estimations of  $\beta^{*}$.
\end{algorithmic}
\end{algorithm}

\section{Computational complexity} 

We look at the number of floating point operations line by line in
Algorithm \ref{alg1}.
Clearly it takes  $O(p)$ flops to finish step 2-4.
In step 5, we use conjugate gradient (CG) method (Golub and Van Loan, 2012) to solve the linear equation iteratively. During  CG iterations
the main operation is two matrix-vector multiplications which cost $2n|A_{k+1}|$ flops  (the term $\bX^{\prime}y$ on the right-hand side can be precomputed and stored).
Therefore we can control the number of CG iterations smaller than $p/(2|A_{k+1}|)$ to
ensure that  $O(np)$ flops will be enough for  step 5.
In step 6, calculating the matrix-vector product costs  $np$ flops.
As for step 7, checking  the stop condition needs $O(p)$ flops.
So the the overall cost per iteration of Algorithm \ref{alg1} is $O(np)$.
By Corollary  \ref{cori}  it needs no more than  $O(\log(R))$  iterations to get a good solution   for Algorithm \ref{alg1} under the certain conditions.
Therefore the overall cost of Algorithm \ref{alg1} is  $O(np\log(R))$
for exactly sparse and approximately sparse case under proper conditions.

Now we consider the cost of ASDAR (Algorithm \ref{alg2}).
Assume  ASDAR is stopped  when $k = L$. Then the above discussion shows the the overall cost of Algorithm \ref{alg2} is bounded by
$O(Lnp \log(R))$ which is very efficient for large  scale high dimension problem since the cost increases linearly in the ambient dimension $p$.

\section{Comparison with greedy and screening methods}
\label{Relatedwork}
We give a high level comparison between SDAR and several greedy and screening methods,
including OMP [Mallat and Zhang (1993),  Tropp (2004), Donoho,  Elad and Temlyakov (2006),
Cai and Wang (2011), Zhang (2011a)], FoBa [Zhang 2011b)],  GraDes [Garg and Khandekar (2009)], and SIS [Fan and Lv (2008)].
These greedy methods iteratively select/remove  one or more variables
and project the response vector  onto
the linear subspace spanned by the variables that have already
been selected. From this point of view, they and SDAR share a similar characteristic.
However, OMP and FoBa,  
select one variable per iteration  based on the current correlation, i.e.,  the dual variable  $\dk$ in our notation, while SDAR selects $T$ variables at a time
based on the sum of primal ($\bbk$) and dual ($\dk$) information.  The following  interpretation in a low-dimensional setting with a small noise term may clarify the differences between these two approaches. If $X^{\prime}X/n \approx \textbf{E}$ and
$\eta \approx\textrm{0}$, we have
\begin{equation*}
d^k = \bX^{\prime}(y-\bX\bbk)/n = \bX^{\prime}(\bX\bbp + \eta -\bX\bbk)/n \approx \beta^{*}-\bbk+\bX^{\prime}\eta/n \approx \beta^{*}-\bbk,
\end{equation*}
and
$$\bbk+d^k\approx \beta^{*}.$$
Hence, SDAR can approximate the underlying support $A^{*}$  more accurately
than OMP and Foba. This is supported by the simulation results given in Section 6.
GraDes can be formulated as
\begin{equation}\label{iht}
\beta^{k+1} = H_{K}( \beta^{k} + s_k \dk ),
\end{equation}
where $H_{K}(\cdot)$ is the hard thresholding operator by keeping the first $K$ largest elements and setting others to $0$, $s_k$ is the step size of gradient descent.  Specifically, GraDes uses $s_k = 1/(1+\delta_{2K})$, where $\delta_{2K}$ is the RIP constant. 
Intuitively,  GraDes  works  by reducing the squares loss with gradient descent with different step sizes
and preserving sparsity using hard thresholding. Hence, GraDes uses  both primal and dual information to detect the support of the solution,  which is similar to SDAR. However, after the approximate active set is determined, SDAR does a least-square fitting,
which is more efficient and more accurate than just keeping the largest elements by hard thresholding. This is also supported by the simulation results given in Section 6.

Fan and Lv (2008) proposed SIS for dimension reduction in ultrahigh dimensional liner regression problems.
This method selects variables with the $T$ largest absolute values  of $\bX^{\prime}y$.
To improve the performance of SIS, Fan and Lv (2008) also considered an iterative  SIS, which iteratively
selects more than one feature at a time until a desired number of variables are
selected. They reported  that the iterative SIS outperforms SIS numerically.
However,  the iterative SIS lacks a theoretically analysis.
Interestingly, the first step in SDAR initialized with $\textbf{0}$ is exactly the same as the SIS. But again the process of SDAR is different from the iterative SIS
in that the active set of SDAR is determined based on the sum of primal and dual approximations while the iterative SIS uses dual only.


\section{Simulation Studies}

\subsection{Implementation}
We implemented SDAR/ASDAR, FoBa, GraDes and MCP in MatLab.
For FoBa, our MatLab implementation follows the R package developed by
Zhang (2011a). We optimize it by keeping track of rank-one updates after each greedy step. Our implementation of MCP uses the iterative threshholding
algorithm (She, 2009) with warm start. The publicly available 
 Matlab packages for LARS (included in the SparseLab package)
are used. Since LARS and FoBa add one variable at a time,
we stop them when $K$ variables are selected in addition to their default stopping conditions. (Of course, by doing this it will speed up and get better solutions for these three solvers).

In GraDes, the optimal gradient step length $s_k$ depends on the RIP constant of $\bX$,
which is NP hard to compute [Tillmann and  Pfetsch (2014)]. Here, we set $s_k=1/3$
following Garg and Khandekar (2009).  We stop GraDes when the residual norm is smaller than  $\vps = \sqrt{n}\s$, or the maximum number of iterations is greater than $n/2$.
We compute the MCP solution path and select an optimal solution using the HBIC
[Wang, Kim and Li (2013)].
We stop the iteration when the residual norm is smaller than  $\vps = \normt{\eta}$,
or the estimated support size is greater than $L =n/\log(n)$.
In ASDAR (Algorithm \ref{alg2}), we set $\tau=50$ and we stop the iteration if the residual
$\|y-\bX\beta^k\|$ is smaller than  $\vps = \sqrt{n}\s$ or $k \ge L= n/\log(n)$.

\subsection{Accuracy and efficiency}
We compare the accuracy and efficiency of SDAR/ASDAR with Lasso (LARS),
MCP, GraDes and FoBa.

We first generate an $n \times p$ random Gaussian matrix $\bar{\bX}$  whose entries are i.i.d. $\sim\cN(0, 1)$ and then normalize its columns to the $ \sqrt{n}$ length.
Then the design matrix $\bX$ is generated with $\bX_1=\bar{\bX}_1$ and
$\bX_{j} = \bar{\bX}_{j}+\rho(\bar{\bX}_{j+1}+\bar{\bX}_{j-1}), j=2, \ldots, p-1$. The underlying regression coefficient $\bbp$ is generated with the nonzero coefficients uniformly distributed in $[m,M]$, where $m = \sigma\sqrt{2\log(p)/n}$ and
$M = 100m$. Then the observation vector $y =X\beta^{*}+\eta$ with $\eta_{i}\thicksim \cN(0,\sigma^2)$, $i=1,\ldots, n$.

We consider a moderately large scale setting with $n=5000$ and $p=50000$. 
The number of nonzero coefficients is set to be $K =400$.
So the sample size $n$ is about
$O(K\log(p-K))$, which is nearly at the limit of estimating $\bbp$ in the linear model (\ref{model}) by the Lasso with theoretically guaranteed [Wainwright (2009)].
The data are generated from the model described above with $K=400, R= 100$.
We set $\sigma = 1$ and $\rho =0.2, 0.4$ and $0.6$.

Table \ref{timeerror1} shows the results based on $100$ independent replications.
The first column gives the correlation value $\rho$ and
the second column shows the methods in the comparison. The third to fifth  columns give the averaged CPU time (in seconds),
the averaged relative error
defined as ($\textrm{ReErr} = \sum \|\hat{\beta}-\beta^{*}\|/\|\beta^{*}\|$), respectively, where $\hat{\beta}$ denotes the estimates 
and $\hat{A} = \textrm{supp}(\hat{\beta})$.  The standard deviations of the CPU times and  the relative errors are shown in the parentheses.
In each column of the Table \ref{timeerror1}, the results in boldface indicate the best performers.

\begin{table}[H]
\centering
\caption{Numerical results (CPU time, relative errors) on data set with $n=5000,p=50000,K=400,R= 100,\sigma = 1,\rho = 0.2:0.2:0.6$.}\label{timeerror1}
\begin{tabular}{cccccccp{0.3cm}p{0.3cm}c}
\hline\hline
\multicolumn{1}{c}{$\rho$} & \multicolumn{1}{c}{method} &  \multicolumn{1}{c}{ReErr}   & \multicolumn{1}{c}{time(s)} \\
 \hline
                  &LARS            &1.1e-1            (2.5e-2)           &4.8e+1           (9.8e-1) \\
                  &MCP             &\textbf{7.5e-4}	  (\textbf{3.6e-5})  &9.3e+2           (2.4e+3) \\
      $0.2$       &GraDes          &1.1e-3            (7.0e-5)           &2.3e+1           (9.0e-1) \\
                  &FoBa            &\textbf{7.5e-4}	  (7.0e-5)           &4.9e+1           (3.9e-1)\\
                  & ASDAR          &\textbf{7.5e-4}	  (4.0e-5)           &8.4e+0	           (4.5e-1)\\
                  & SDAR           &\textbf{7.5e-4 }  (4.0e-5)           &\textbf{1.4e+0}	   (\textbf{5.1e-2})\\
 \hline                             	
                  &LARS            &1.8e-1            (1.2e-2)           &4.8e+1           (1.8e-1) \\
                  &MCP             &6.2e-4            (3.6e-5)           &2.2e+2           (1.6e+1) \\
         $0.4$    &GraDes          &8.8e-4            (5.7e-5)           &8.7e+2           (2.6e+3) \\
                  &FoBa            &1.0e-2            (1.4e-2)           &5.0e+1           (4.2e-1)\\
                  & ASDAR          &\textbf{6.0e-4}   (\textbf{2.6e-5})  &8.8e+0	           (\textbf{3.2e-1})\\
                  & SDAR           &\textbf{6.0e-4 }  (\textbf{2.6e-5})  &\textbf{2.3e+0}	   (1.7e+0)   \\
 \hline
                  &LARS            &3.0e-1            (2.5e-2)           &4.8e+1           (3.5e-1) \\
                  &MCP             &4.5e-4            (\textbf{2.5e-5})  &4.6e+2           (5.1e+2) \\
  $0.6$         &GraDes            &7.8e-4            (1.1e-4)           &1.5e+2           (2.3e+2) \\
                  &FoBa            &8.3e-3            (1.3e-2)           &5.1e+1           (1.1e+0)\\
                  & ASDAR          &\textbf{4.3e-4}   (3.0e-5)           &1.1e+1           (5.1e-1)\\
                  & SDAR           &\textbf{4.3e-4}   (3.0e-5)           &\textbf{2.1e+0}    (\textbf{8.6e-2})\\
\hline\hline
\end{tabular}
\end{table}

We see that when the correlation $\rho$ is low, i.e., $\rho = 0.2$, MCP, FoBa, SDAR and ASDAR are on the top of the list in average error (ReErr). In terms of speed,  SDAR/ASDAR is almost 20-900/3-100 times faster than the other methods.
As the correlation $\rho$ increases to $\rho = 0.4$ and $\rho=0.6$,
FoBa becomes less accurate  than SDAR/ASDAR. The accuracy of MCP is similar to that of SDAR/ASDAR, but MCP is 20 to 100 times slower than SDAR/ASDAR.
The standard deviations of the CPU times and  the relative errors of MCP and SDAR/ASDAR are similar and smaller than those of the other methods in all the three settings.



\subsection{Influence of the model parameters}
We now consider the effects of each of the model parameters on the performance
of ASDAR, LARS, MCP, GraDes and FoBa more closely.

\begin{figure}[H]
  \centering
  \begin{tabular}{cc}
    \includegraphics[trim = 0.5cm 0.5cm 0.5cm 0.5cm, clip=true,width=5cm]{{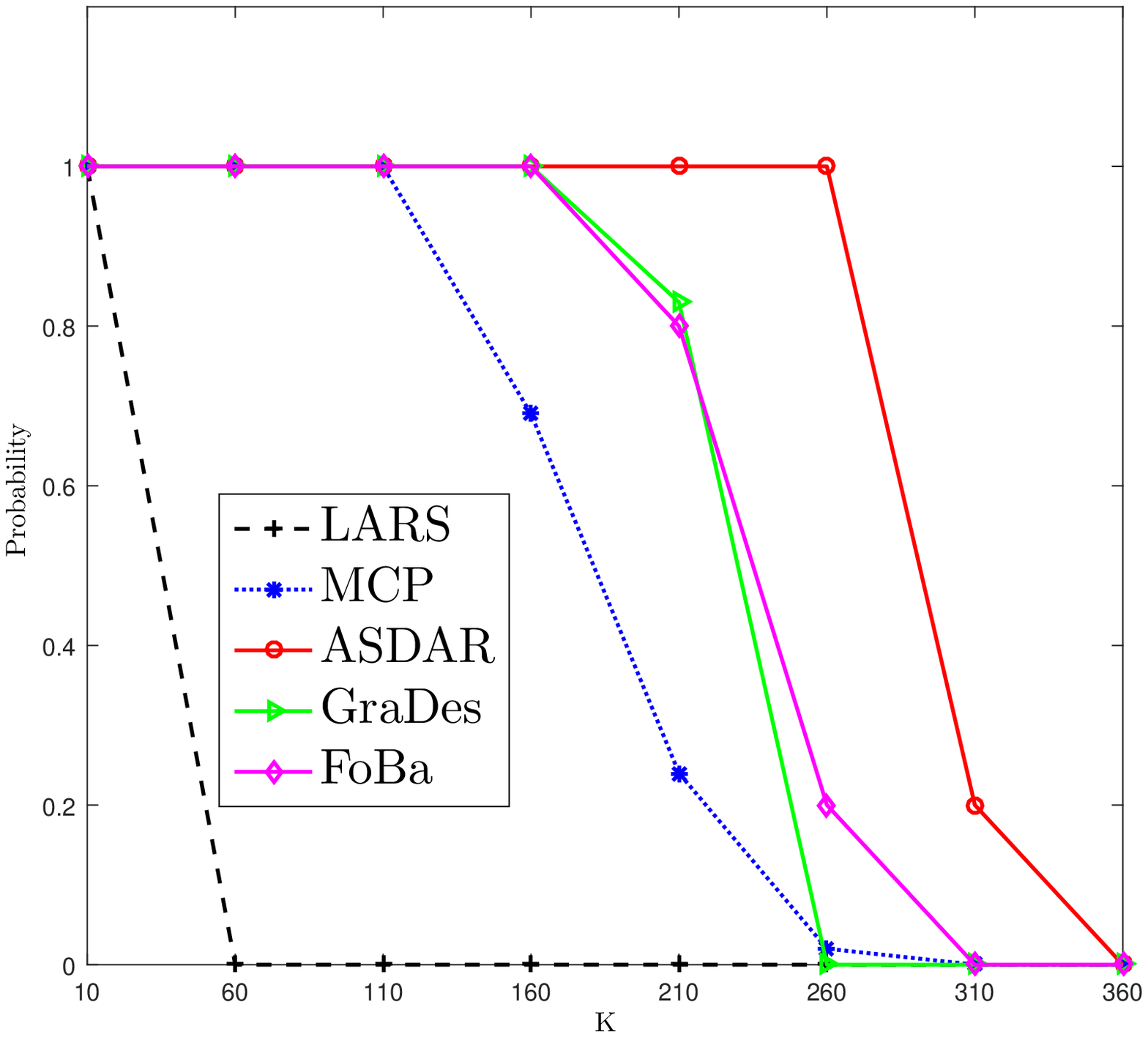}} &\includegraphics[trim = 0.5cm 0.5cm 0.5cm 0.5cm, clip=true,width=5cm]{{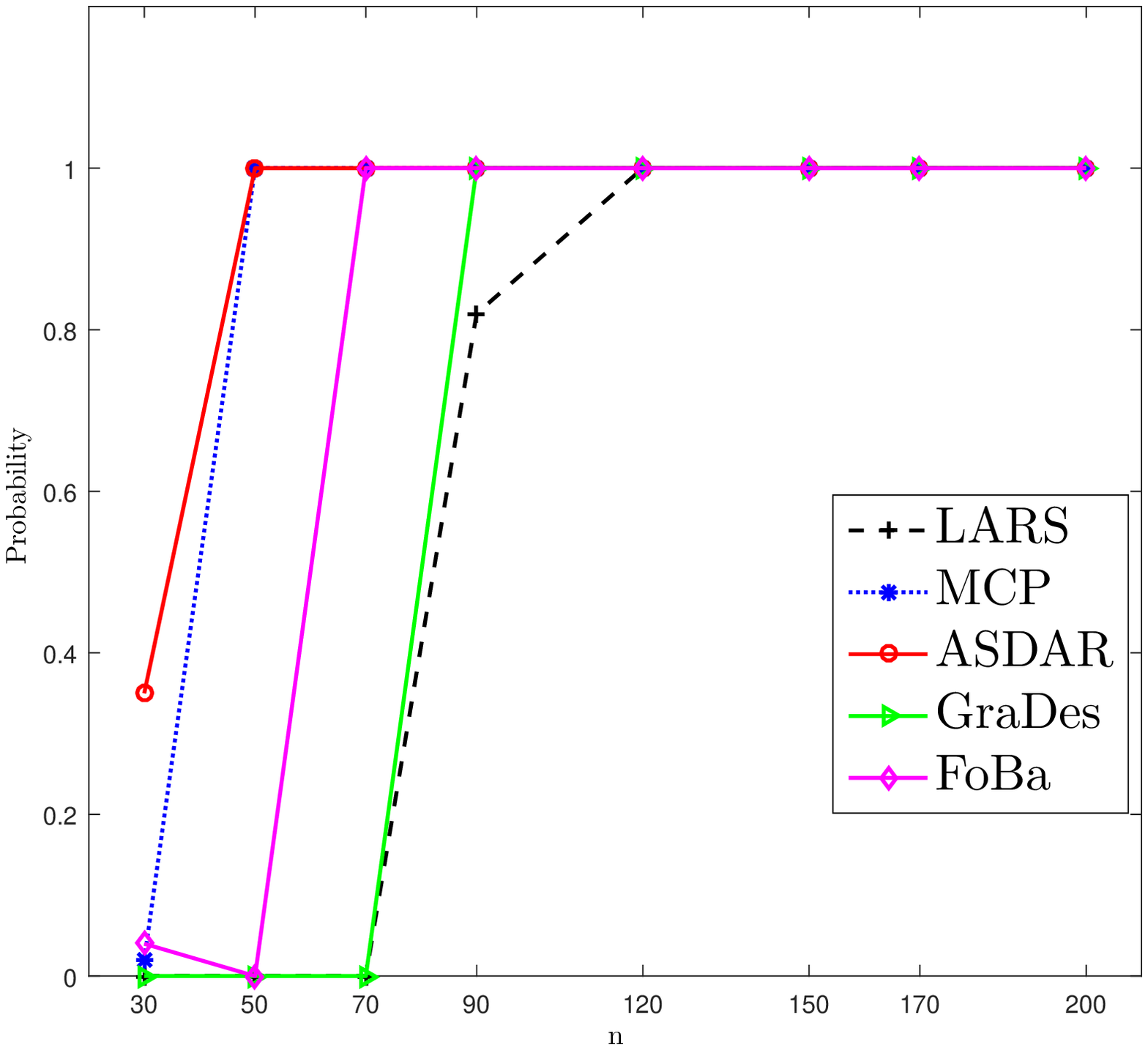}}\\
    \includegraphics[trim = 0.5cm 0.5cm 0.5cm 0.5cm, clip=true,width=5cm]{{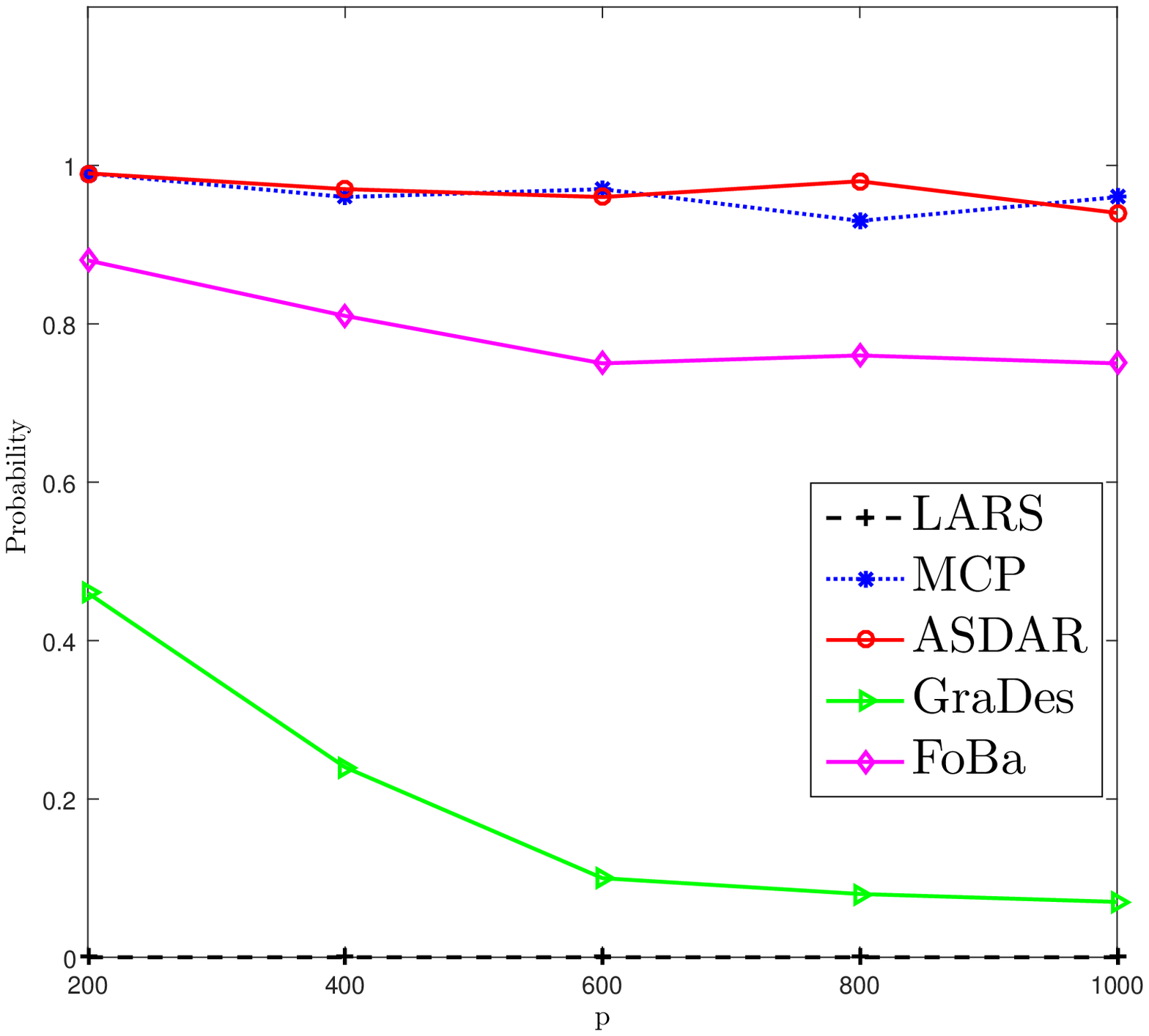}} &\includegraphics[trim = 0.5cm 0.5cm 0.5cm 0.5cm, clip=true,width=5cm]{{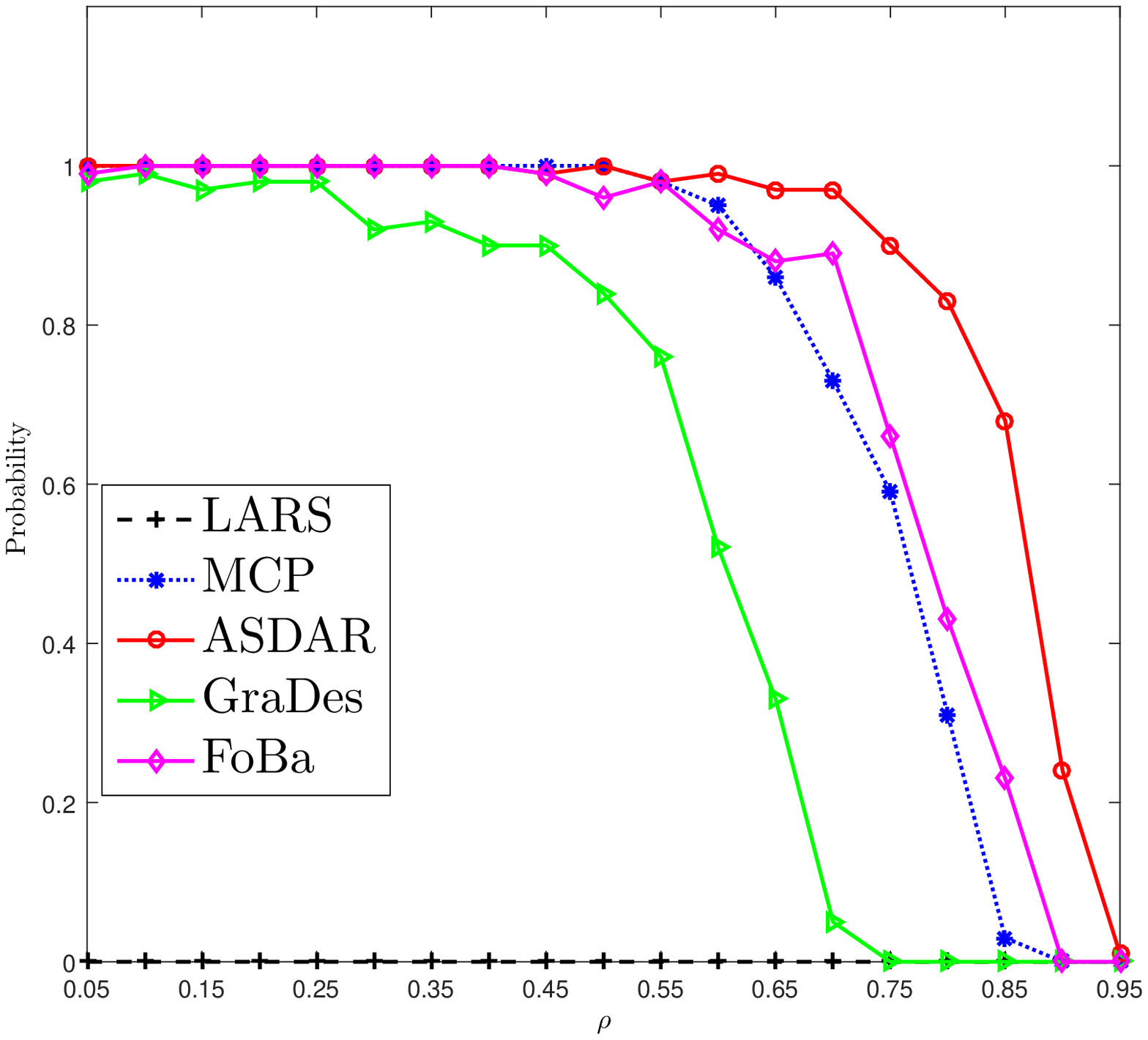}}\\
  \end{tabular}
  \caption{Numerical results of the influence of sparsity level $K$ (top left panel), sample size $n$  (top right panel), ambient dimension $p$ (bottom left  panel)  and correlation $\rho$  (bottom  right panel) on the probability of exact recovery of the true support  of all the  solvers considered  here.
  }\label{influenceofparameter}
\end{figure}

In this set of simulations, the rows of the design matrix $X$
are drawn independently from $\cN(0,\Sigma)$  with $\Sigma_{jk} = \rho^{|j-k|}, 1 \le j, k \le p$. The elements of the error vector $\eta$ are generated independently with
$\eta_{i}\thicksim \cN(0,\sigma^2)$, $i=1,\ldots, n$.
Let $R = M/m$,
where, $M = \max\{|\beta^{*}_{A^{*}}|\}, m= \min\{|\beta^{*}_{A^{*}}|\}=1$.
The underling regression coefficient vector $\beta^{*} \in \mathbb{R}^{p}$ is generated in such a way that $A^{*}$ is a randomly chosen subset of $\{1,2,...,p\}$ with $|A^{*}|=K< n$  and $R\in [1,10^3]$.
Then the observation vector $y =X\beta^{*}+\eta$. We use
$\{n, p, K, \sigma, \rho, R\}$
to indicate the parameters used in the data generating model described above.      		

We run ASDAR with $\tau=5, L = n/\log(n)$ (if not specified).
We use the HBIC
[Wang, Kim and Li (2013)] to select the tuning parameter $T$.
The simulation results given in Figure \ref{influenceofparameter} are based on 100 independent replications.

\subsubsection{Influence of the sparsity level $K$}
The top left panel of Figure \ref{influenceofparameter} shows the results of the influence of sparsity level $K$ on the probability
of exact recovery of $A^{*}$ of  ASDAR, LARS, MCP, GraDes and FoBa on data sets
with $(n=500, p =1000, K = 10:50:360, \sigma = 0.5, \rho = 0.1, R = 10^3)$.   Here $10:50:360$ means the sample size starts from 10 to 360 with an increment of 50.  We use $L=0.8n$ for both ASDAR and MCP to eliminate the effect of stopping rule since the maximum $K=360$.
When the sparsity level $K=10$, all the solvers performed well in 
recovering the true support. As $K$ increases, LARS was the first one that failed to recover the support and vanished when $K=60$ (this phenomenon had also been observed in Garg and  Khandekar (2009)), MCP began to fail when $K>110$,  GraDes and FoBa began to fail when $K>160$. In comparison, ASDAR was still able to do well even when $K=260$.



\subsubsection{Influence of the sample size $n$}
The top right panel of Figure \ref{influenceofparameter} shows the influence of the
sample size $n$ on the probability  
of correctly estimating 
$A^{*}$ 
with data generated from the model with
$(n=30:20:200, p =500, K =10, \sigma=0.1, \rho =0.1, R=10)$.
We see that the performances of all the five methods become better as $n$ increases.
However, ASDAR performs better than the others when $n = 30$ and $50$.


\subsubsection{Influence of the ambient dimension $p$}
The bottom  left panel of Figure \ref{influenceofparameter}  shows  the influence of ambient dimension $p$ on the performance of ASDAR, LARS, MCP, GraDes and FoBa on data with
$(n=100, p =200:200:1000, K = 20, \sigma = 1, \rho = 0.3, R = 10)$. We see that the probabilities of exactly recovering the support of the underlying coefficients of ASDAR and MCP are higher than those of the other solvers as
$p$ increasing, which indicate that ASDAR and MCP are more robust to the ambient dimension.

\subsubsection{Influence of correlation $\rho$}
The bottom right panel of Figure \ref{influenceofparameter}
shows the influence of correlation $\rho$ on the performance of ASDAR, LARS, MCP, GraDes and FoBa on data with $(n=150, p =500, K = 25, \sigma = 0.1, \rho=0.05:0.1:0.95, R = 10^2)$.   We see that the performance of all the solvers become worse  when the correlation $\rho$ increasing. However, ASDAR generally performed better than the other methods
as $\rho$ increases.



In summary, our simulation studies demonstrate that SDAR/ASDAR is generally more accurate, more efficient and more stable than Lasso, MCP, FoBa and GraDes.


\section{Concluding remarks}
SDAR is a constructive approach for fitting sparse,  high-dimensional linear regression models. Under appropriate conditions, we established the nonasymptotic minimax $\ell_2$ error bound and optimal
$\ell_{\infty}$ error bound of the solution sequence generated
by SDAR. We also calculated the number of iterations required to achieve these bounds.
In particular, an interesting and surprising aspect of our results is,
if a mutual coherence condition on the design matrix is satisfied,
the number of iterations required for the SDAR to achieve the optimal $\ell_{\infty}$
bound does not depend on the underlying sparsity level. In addition, SDAR has the same computational complexity per iteration as LARS, coordinate descent and greedy methods.
 SDAR/ ASDAR is accurate, fast, stable and easy to implement. Our simulation studies demonstrate that SDAR/ ASDAR is competitive with or outperforms the Lasso, MCP and several greedy methods in efficiency and accuracy. These theoretical and numerical results suggest that SDAR/ ASDAR is a   promising alternative to the existing penalized and greedy methods
for dealing with large scale high-dimensional linear regression problems.


We have only considered the linear regression model.
It would be interesting to generalize SDAR to
other models with more general loss functions or models
with other types of sparsity structures.
It would also be interesting to develop parallel or distributed versions of SDAR that can run
on multiple cores for data sets with big $n$ and large $p$ or when data are
distributively stored.


We have implemented SDAR in a Matlab package \textsl{sdar}, which is available at
\url{http://homepage.stat.uiowa.edu/~jian/}.

\section*{Acknowledgments}

We are grateful to two anonymous reviewers, the Associate Editor and Editor for their helpful comments which led to considerable improvements in the paper.


\section{Appendix: Proofs}

\noindent
\textbf{Proof of Lemma \ref{lem1}.}
Let $L_{\la}(\beta)=\frac{1}{2n}\normt{\bX\bb-y}^2 + \la\normz{\bb}.$
Suppose $\bbs$ is a  minimizer of $L_{\la}$. Then
\begin{equation*}
\begin{aligned}
    &\beta_i^\diamond \in \mathop\textrm{argmin}\limits_{t\in \mathbb{R}} L_{\la}(\beta_1^\diamond,...,\beta_{i-1}^\diamond,t,\beta_{i+1}^\diamond,...,\beta_p^\diamond)\\
    \Rightarrow\quad & \beta_i^\diamond \in \mathop\textrm{argmin}\limits_{t\in \mathbb{R}} \tfrac{1}{2n}\normt{\bX \bbs-y+(t-\beta_i^\diamond)\bX_i}^2+\la|t|_{0} \\
    \Rightarrow\quad & \beta_i^\diamond \in \mathop\textrm{argmin}\limits_{t\in \mathbb{R}} \tfrac{1}{2}(t-\beta_i^\diamond)^2+\tfrac{1}{n}(t-\beta_i^\diamond)\bX_i^{\prime}(\bX\bbs-y) + \la|t|_{0} \\
   \Rightarrow\quad & \beta_i^\diamond \in \mathop\textrm{argmin}\limits_{t\in \mathbb{R}} \tfrac{1}{2}(t-(\beta_i^\diamond+\bX_i^{\prime}(y-\bX\bbs)/n))^2 + \la|t|_{0}.
  \end{aligned}
\end{equation*}
 Let $d^\diamond = \bX^{\prime} ( y -\bX\bbs)/n$.  By
the definition of the hard thresholding operator $H_{\la}(\cdot)$ in (\ref{hardth}), we have
\begin{equation*}
   \bbs_{i}=H_{\la}  (\bbs_{i} + d^\diamond_{i}) \quad \mbox{ for } i=1,...,p,
\end{equation*}
which shows \eqref{eq2} holds.

Conversely, suppose \eqref{eq2} holds.
Let \begin{equation*}
\As = \dkhB{ i \in S \big| \abs{\bbs_i + \ds_i} \geq \sqrt{2\la}}.
\end{equation*}
By \eqref{eq2} and the definition of $H_{\la}(\cdot)$ in (\ref{hardth}), we deduce that for $i\in \As$, $|\beta_i^\diamond|\geq \sqrt{2\la}$. Furthermore,
    $\textbf{0}= d^\diamond_{\As} = \bX_{\As}^{\prime} (y - \bX_{\As} \bbs_{\As})/n$,
which is equivalent to
\begin{equation}\label{dual}
\bbs_{\As} \in \mathop\textrm{argmin} \tfrac{1}{2n}\normt{\bX_{\As} \beta_{\As} - y}^2.
\end{equation}
Next we show $L_{\la}(\bbs+ h) \geq L_{\la}(\bbs)$  if  $h$ is small enough with $\normi{h}<\sqrt{2\la}$. Two cases should be considered.
If $h_{(\As)^c}\neq 0$, then
\begin{equation*}
   \begin{aligned}
    L_{\la}(\bbs+ h) - L_{\la}(\bbs) &\geq \tfrac{1}{2n}\normt{\bX \bbs- y + \bX h}^2 -\tfrac{1}{2n}\normt{\bX \bbs- y}^2  + \la
      \geq \la - |\langle h, \ds \rangle|,
   \end{aligned}
\end{equation*}
which is positive for sufficiently small $h$. If $h_{(\As)^c} = 0$,
by the minimizing property of $\bbs_{\As}$ in (\ref{dual}) we deduce $L_{\la}(\bbs + h) \geq L_{\la}(\bbs)$.
This completes the proof of Lemma  \ref{lem1}. $\hfill\Box$

\begin{lem}\label{lem2}
Let $A$, $B$ be disjoint subsets of $S$, with $\abs{A}=a$, $\abs{B}=b$. Assume  $\bX \sim \textrm{SRC}(\cm(a+b), \cp(a+b),a+b)$.
Let $\theta_{a,b}$ be the sparse orthogonality constant
and let $\mu$ be the mutual coherence of $\bX$.
Then we have
\begin{align}
&n\cm(a) \leq\norm{\bX_{A}^T\bX_A }\leq n\cp(a),\label{lem2-1}\\
&\frac{1}{n\cp(a)} \leq \norm{(\bX_A^T\bX_A)^{-1}} \leq \frac{1}{n\cm(a)},\label{lem2-2} \\
& \norm{\bX_A^{\prime}} \leq \sqrt{n\cp(a)}\label{lem2-3}\\
& \theta_{a,b}  \leq (\cp(a+b)-1)\vee(1-\cm(a+b))\label{lem2-4}\\
&\normi{\bX_{B}^{\prime}\bX_{A}u}  \leq n a\mu\normi{u}, \quad \forall  u \in \mathbb{R}^{\abs{A}},\label{lem2-5}\\
&\norm{\bX_{A}}=\norm{\bX^{\prime}_{A}}\leq \sqrt{n(1 + (a-1)\mu)}. \label{lem2-6}
\end{align}
Furthermore, if  $ \mu<1/(a-1)$,   then
\begin{equation}\label{lem2-7}
\normi{(\bX_{A}^{\prime}\bX_{A})^{-1}u} \leq\frac{\normi{u}}{ n (1-(a-1)\mu)}, \quad \forall  u \in \mathbb{R}^{\abs{A}}.
\end{equation}
Moreover,  $\cp(s)$ is an increasing function of $s$, $\cm(s)$ a decreasing function of $s$ and $\theta_{a,b}$ an increasing function of $a$ and $b$.
\end{lem}

\medskip\noindent
\textbf{Proof of Lemma \ref{lem2}.}
The  assumption  $\bX \sim \textrm{SRC}(a, \cm(a), \cp(a))$  implies the spectrum of  $\bX_{A}^{\prime}\bX_{A}/n$ is contained in  $[\cm(a),\cp(a)]$. So
 (\ref{lem2-1}) - (\ref{lem2-3})  hold.  Let $\textbf{E}$ be an $(a+b)\times(a+b)$ identity matrix.  (\ref{lem2-4}) follows from the fact that  $\bX_A^{\prime} \bX_B/n$ is a submatrix of $\bX_{A\cup B}^{\prime} \bX_{A\cup B}/n- \textbf{E}$ whose spectrum norm is less than $(1-\cm(a+b))\vee (\cp(a+b)-1)$.  Let $G =\bX^{\prime}\bX/n$.  Then,   $|\sum_{j=1}^{a}G_{i,j}u_{j}|\leq\mu a\normi{u},$ for all $i\in B$, which implies (\ref{lem2-5}).   By Gerschgorin's disk theorem,
$$ |\norm{G_{A,A}}-G_{i,i}|\leq \sum_{i\neq j=1}^{a}|G_{i,j}|\leq  (a-1)\mu \quad \forall i\in A,$$
thus (\ref{lem2-6}) holds. For (\ref{lem2-7}), it suffices to show $\normi{G_{A,A}u}\geq(1-(a-1)\mu)\normi{u}$ if $\mu< 1/(a-1)$.
In fact, let $i\in A$ such that $\normi{u} = \abs{u_i}$, then $$\normi{G_{A,A}u}\geq |\sum_{j=1}^{a}G_{i,j}u_{j}|\geq|u_i| - \sum_{i\neq j=1}^{a}\abs{G_{i,j}}|u_{j}| \geq \normi{u} - \mu (a-1)\normi{u}.$$
The last assertion follows from their definitions.
This completes the proof of Lemma  \ref{lem2}.  $\hfill\Box$

\medskip
We now define some notation that will be useful in the proof of
Theorems \ref{thmt} and \ref{thmi} given below.
For any given integers  $T$ and $J$ with $T\geq J$ and  $F \subseteq S$ with $|F| = T -J$, let $A^{\circ} =  A^{*}_J \cup F$ and  $I^{\circ} = (A^{\circ})^c$. Let $\{\Ak\}_k$ be the active sets generated by Algorithm \ref{alg1}.
  Define
\[\JAk = \normt{\tb|_{A^{*}_J\backslash \Ak}} \ \mbox{ and } \
\JAki = \normi{\tb|_{A^{*} \backslash \Ak}}.
\]
Let
\[
\Ak_1 = \Ak \cap A^{\circ},
  \Ak_2 = A^{\circ} \backslash \Ak_1, \Ik_3 = \Ak \cap I^{\circ}, \Ik_4 = I^{\circ} \backslash \Ik_3.
\]
Denote the cardinality of $\Ik_3$ by $l_{k}=|\Ik_3|$.
Let
\[
\Ak_{11} = \Ak_1 \backslash (\Akk \cap \Ak_1), \Ak_{22} = \Ak_2 \backslash (\Akk \cap \Ak_2), \Ik_{33} = \Akk \cap \Ik_3,
 \Ik_{44} = \Akk \cap \Ik_4,
\]
and $$\bbb = \bbkk - \tb|_{\Ak}.$$
These notation can be easily understood in the case $T=J$.
For example,
 $\JAk$ ($\JAki$) is a measure of the difference between the detected active set  $\Ak$ and the target support $A^{*}_J$. $\Ak_1$ and $\Ik_3$ contain the correct indexes  and incorrect indexes  in $\Ak$, respectively. $\Ak_{11}$ and $\Ak_{22}$ include the  indexes  in $A^{\circ}$ that will be lost from the $k$th to $(k+1)$th iteration.
$\Ik_{33}$ and $\Ik_{44}$ contain the indexes included in $I^{\circ}$ that  will be gained.
By Algorithm \ref{alg1}, we have
$|\Ak|= |\Akk| = T$, $\Ak = \Ak_1 \cup \Ik_3,$
   $|\Ak_2| = |A^{\circ}| - |\Ak_1|= |A^{\circ}| - |\Ik_3|=T-(T-l_{k}) = l_{k}\leq T$, and
  \begin{align}
  |\Ak_{11}| + |\Ak_{22}| &= |\Ik_{33}| + |\Ik_{44}|,\label{lemequl}\\
  \JAk &= \normt{\tb|_{A^{\circ} \backslash \Ak}}=\normt{\tb|_{\Ak_2}}, \label{Jkeq}\\
  \JAki &= \normi{\tb|_{A^{\circ} \backslash \Ak}}=\normi{\tb|_{\Ak_2}}. \label{Jkeqinf}
   \end{align}

Before we give the technical proofs of  Theorems and Corollaries we give  description of of the main ideas behind the proofs.
Intuitively, SDAR is
a support detection and least square fitting process. Therefore  our proofs  justify the active sets $\{\Ak\}_k$ generated by Algorithm \ref{alg1} can
approximate $A^{*}_J$ more and more accurately by showing that $\JAk, \JAki$ decays geometrically
and the effect of the noise $\te$ can be well controlled with high probability.
 To this end, we need  the following technical Lemmas \ref{lem4} - \ref{lem9}.
 Lemma \ref{lem4} shows the effect of noise  $\htt$ and $\hit$  can be controlled by the sum of unrecoverable energy $R_{J}$ and the universal noise level $O(\s\sqrt{2\log(p)/n})$  with high probability if   $\eta$ is sub-Gaussian.
Lemma \ref{lem5} shows the norm of both  $\bbb$  and  $\bbk-\tb$ are mainly bounded by $\JAk$ and $\htt$ ($\JAki$ and $\hit$).
Lemma \ref{lem6} shows that $\JAkk$ ($\JAkki$)  can be bounded by  the norm of $\tb$    on the lost indexes    and further  can be mainly controlled  in terms of   $\JAk$, $\htt$ ($\JAki$, $\hit$) and the norm  of $\bbb$, $\bbkk$  and  $\dkk$ on the lost indexes.
Lemma \ref{lem7} gives the benefits brought by the orthogonality  of $\bbk$ and $d^k$
that the    norm of  $\bbkk$  and  $\dkk$ on the lost indexes    can be bounded  by the  norm  on the  gained indexes.
Lemma \ref{lem8} gives the upper bound of the    norm of $\bbkk+\dkk$ on the gained indexes by the sum of   $\JAk$,  $\htt$ ($\JAki$, $\hit$), and the norm of $\bbb$.
 Then Lemma \ref{lem9} get the desired relations  of  $\JAkk$ and $\JAk$ ($\JAkki$ and $\JAki$) by combining Lemma \ref{lem4} - \ref{lem8}.

\begin{lem}\label{lem3}
Suppose (A3) holds. We have for any $\alpha\in(0,1/2)$
\begin{align}\label{lem3-1}
& \bP\xkhB{\normi{\bX^{\prime} \eta} \leq  \s\sqrt{2\log(p/\alpha)n}}\geq 1-2\alpha \\
\label{lem3-2}
&\bP\xkhB{\max_{|A|\le T} \normt{\bX_A' \eta}\leq  \s \sqrt{T} \sqrt{2\log(p/\alpha)n}}
\geq 1-2\alpha.
\end{align}
\end{lem}
\medskip\noindent
\textbf{Proof of Lemma \ref{lem3}.}
This lemma follows from the sub-Gaussian assumption (A3) and standard probability calculation, see  Cand\`{e}s and Tao (2007), Zhang and Huang (2008), Wainwright (2009) for a detail.
$\hfill\Box$

\begin{lem}\label{lem4}
Let  $A \subset S$ with $|A| \leq T$.
Suppose (A1) and (A3) holds.  Then for $\alpha\in (0,1/2)$ with probability at least $1 - 2\alpha$, we have

\begin{itemize}
\item[(i)]
If  $\bX\sim SRC(T, \cm(T), \cp(T))$,
then
\begin{equation}\label{lem4-1}
\htt\leq \vps_1.
\end{equation}
\item[(ii)]
\begin{equation}\label{lem4-2}
\hit\leq \vps_2.
\end{equation}
\end{itemize}
\end{lem}
\begin{proof}
We first show
\begin{equation}\label{nonexpan}
\normt{\bX \bbp|_{(A^{*}_{J})^c}} \leq \sqrt{n\cp(J)}R_{J},
\end{equation}
under the assumption of  $\bX\sim SRC(\cm(T), \cp(T),T)$ and (A1).
In fact,
let $\beta$ be an arbitrary vector in $\bRp$  and $A_{1}$ be the first $J$ largest positions of $\beta$, $A_{2}$ be the next and so forth.
Then
\begin{align*}
\normt{\bX\beta}&\leq \normt{\bX\beta_{A_{1}}} + \sum_{i\geq2}\normt{\bX\beta_{A_{i}}}\\
& \leq\sqrt{n\cp(J)}\normt{\beta_{A_{1}}}+\sqrt{n\cp(J)} \sum_{i\geq2}\normt{\beta_{A_{i}}}\\
&\leq \sqrt{n\cp(J)}\normt{\beta} + \sqrt{n\cp(J)} \sum_{i\geq1}\sqrt{\frac{1}{J}}\normo{\beta_{A_{i-1}}}\\
&\leq \sqrt{n\cp(J)} (\normt{\beta}+\sqrt{\frac{1}{J}}\normo{\beta}),
\end{align*}
where the first inequality uses the triangle inequality, the second    inequality uses (\ref{lem2-3}), and the third and fourth ones are simple algebra.
This implies (\ref{nonexpan}) holds by observing the definition of $R_{J}$.
By the triangle inequality, (\ref{lem2-3}), (\ref{nonexpan}) and  (\ref{lem3-2}),
we have  with probability at least $1-2\alpha$,
\begin{align*}
 \normt{\bX_A^{\prime}\te}/n&\leq \normt{\bX^{\prime}_A\bX \bbp|_{(A^{*}_{J})^c}}/n+\normt{\bX^{\prime}_A\eta}/n\\
  &\leq \cp(J)R_{J} + \s\sqrt{T}\sqrt{2\log(p/\alpha)/n},
 \end{align*}
 Therefore, \eqref{lem4-1} follows  by noticing the monotone increasing  property of $\cp(\cdot)$, the definition of $\vps_1$ and the arbitrariness of $A$.

 Repeating the proof process of (\ref{nonexpan2}) and replacing  $\sqrt{n\cp(J)}$ with $\sqrt{n(1+(J-1)\mu)}$  by using (\ref{lem2-6}) we get
 \begin{equation}\label{nonexpan2}
\normt{\bX \bbp|_{(A^{*}_{J})^c}} \leq \sqrt{n(1+(K-1)\mu)}R_{J}.
\end{equation}
Therefore, by (\ref{lem2-6}), (\ref{nonexpan2}) and  (\ref{lem3-1}), we have with probability at least  $ 1-2\alpha$,
\begin{align*}
 \normi{\bX_A^{\prime}\te}/n&\leq \normi{\bX_A^{\prime}\bX \bbp|_{(A^{*}_{J})^c}}/n+\normt{\bX^{\prime}_A\eta}/n\\
 &\leq \normt{\bX_A^{\prime}\bX \bbp|_{(A^{*}_{J})^c}}/n+\normt{\bX^{\prime}_A\eta}/n\\
  &\leq (1+(J-1)\mu) R_{J} + \s\sqrt{2\log(p/\alpha)/n}.
 \end{align*}
   This implies  part (ii) of  Lemma \ref{lem4} by noticing the definition of $\vps_2$ and the arbitrariness of $A$.
\end{proof}

\begin{lem}\label{lem5}
Let  $ A \subset S$ with $|A| \leq T$.
Suppose (A1) holds.
\begin{itemize}
\item[(i)] If $\bX\sim \textrm{SRC}(T, \cm(T), \cp(T))$,
\begin{equation}
\normt{\bbkk - \tb} \leq \xkhB{1 + \frac{\theta_{T,T}}{\cm(T)}} \JAk
+ \frac{\htt}{\cm(T)},  \label{lem5-1}
\end{equation}
\item[(ii)] If $(T-1)\mu<1$, then
\begin{equation}
\normi{\bbkk - \tb} \leq \frac{1+\mu}{1-(T-1)\mu} \JAki + \frac{\hit}{1-(T-1)\mu}. \label{lem5-1'}
\end{equation}
\end{itemize}
\end{lem}

\begin{proof}
\begin{align}
\bbkk_{\Ak} &= (\bX_{\Ak}^{\prime}\bX_{\Ak})^{-1}\bX_{\Ak}^{\prime}y\notag\\
&= (\bX_{\Ak}^{\prime}\bX_{\Ak})^{-1}\bX_{\Ak}^{\prime}(\bX_{\Ak_1}\tb_{\Ak_1} + \bX_{\Ak_2}\tb_{\Ak_2} + \te),\label{lem5-2} \\
(\tb |_{\Ak})_{\Ak} &= (\bX_{\Ak}^{\prime}\bX_{\Ak})^{-1}\bX_{\Ak}^{\prime}\bX_{\Ak}(\tb|_{\Ak})_{\Ak}\notag\\
&= (\bX_{\Ak}^{\prime}\bX_{\Ak})^{-1}\bX_{\Ak}^{\prime}(\bX_{\Ak_1}\tb_{\Ak_1})\label{lem5-3}
\end{align}
where the first equality uses the definition of $\bbkk$ in Algorithm \ref{alg1}, the second equality uses    $y = \bX\tb + \te = \bX_{\Ak_1}\tb_{\Ak_1} + \bX_{\Ak_2}\tb_{\Ak_2} + \te$, the third equality  is simple algebra, and the last one uses the definition of
$\Ak_1.$
\begin{align}
\normt{\bbb}& = \normt{\bbkk_{\Ak} - (\tb|_{\Ak})_{\Ak}}   \notag\\
 &=\normt{(\bX_{\Ak}^{\prime}\bX_{\Ak})^{-1}\bX_{\Ak}^{\prime}(\bX_{\Ak_2}\tb_{\Ak_2} + \te)} \notag\\
&\leq \frac{1}{n\cm(T)} (\normt{\bX_{\Ak}^{\prime}\bX_{\Ak_2}\tb_{\Ak_2}} + \normt{\bX_{\Ak}^{\prime}\te}) \notag\\
&\leq  \frac{\theta_{T,T}}{\cm(T)} \normt{\tb|_{\Ak_2}} +  \frac{\htt}{\cm(T)} \label{lem5-4}
\end{align}
where the first equality uses $\textrm{supp}(\bbkk) = \Ak$, the second equality uses \eqref{lem5-3} and  \eqref{lem5-2},  the
first inequality uses \eqref{lem2-2} and the triangle inequality,  and the second inequality uses \eqref{Jkeq}, the definition  of $\theta_{a,b}$ and $\htt$.
Then the triangle inequality
$\normt{\bbkk - \tb} \leq \normt{\bbkk - \tb|_{\Ak}} + \normt{\tb|_{A^{\circ} \backslash \Ak}}$ and \eqref{lem5-4}
imply \eqref{lem5-1}.

Similar to the proof of \eqref{lem5-4}  and using \eqref{lem2-7}, \eqref{lem2-5}, \eqref{Jkeqinf}, we have
\begin{align}
\normi{\bbb} \leq  \frac{T\mu}{1-(T-1)\mu} \normi{\tb|_{\Ak_2}} + \frac{\hit}{(1-(T-1)\mu)} \label{lem5-4'}
\end{align}
Thus  \eqref{lem5-1'} follows by using triangle inequality and \eqref{lem5-4'}.
This completes the proof of Lemma  \ref{lem5}.
\end{proof}

\begin{lem}\label{lem6}
\begin{align}
\JAkk &\leq \normt{\tb_{\Ak_{11}}} + \normt{\tb_{\Ak_{22}}},\label{lem6-1}\\
\JAkki &\leq \normi{\tb_{\Ak_{11}}} + \normi{\tb_{\Ak_{22}}}.\label{lem6-1'}\\
\normt{\tb_{\Ak_{11}}}&\leq \normt{\bbb_{\Ak_{11}}} +\normt{\bbkk_{\Ak_{11}}}, \label{lem6-2}\\
\normi{\tb_{\Ak_{11}}}&\leq \normi{\bbb_{\Ak_{11}}} +\normi{\bbkk_{\Ak_{11}}}\label{lem6-2'}.
\end{align}
Furthermore, assume  (A1) holds. We have
\begin{align}
\normi{\tb_{\Ak_{22}}} & \leq \normi{\dkk_{\Ak_{22}}}+ T\mu\normi{\bbb_{\Ak}} + T\mu\JAki + \hit\label{lem6-3'}\\
\normt{\tb_{\Ak_{22}}}&\leq \frac{\normt{\dkk_{\Ak_{22}}} + \theta_{T,T} \normt{\bbb_{\Ak}}+ \theta_{T,T}\JAk
+ \htt}{\cm(T)} \label{lem6-3} \quad\textrm{ if} \quad \bX\sim \textrm{SRC}(\cm(T), \cp(T),T).
\end{align}
\end{lem}
\begin{proof}
\begin{align*}
\JAkk = \normt{\tb|_{A^{\circ} \backslash \Akk}}=
\normt{\tb|_{\Ak_{11}\cup\Ak_{22}}}\leq \normt{\tb_{\Ak_{11}}} + \normt{\tb_{\Ak_{22}}},
\end{align*}
where the first and second equality use the diminution of $\JAkk$ and the definition of $\Ak_{11}$, $\Ak_{11}$, $\Ak_{22}$, respectively.
This proves \eqref{lem6-1}.
\eqref{lem6-1'} can be proved similarly.
\begin{align*}
\normt{\bbkk_{\Ak_{11}}} &= \normt{\xkhB{\tb\big|_{ \Ak}}_{\Ak_{11}} + \bbb_{\Ak_{11}}}
\geq \normt{\tb_{\Ak_{11}}} - \normt{\bbb_{\Ak_{11}}}
\end{align*}
where the equality uses the definition of $\bbb=\bbkk - \tb|_{\Ak}$, the inequality is triangle inequality.
This proves  (\ref{lem6-2}).  (\ref{lem6-2'}) can be proved in the same way.
\begin{align*}
\normt{\dkk_{\Ak_{22}}} &= \normt{\bX_{\Ak_{22}}^{\prime}\xkhB{\bX_{\Ak}\bbkk_{\Ak} - y}/n}\\
&= \normt{\bX_{\Ak_{22}}^{\prime}\xkhB{\bX_{\Ak}\bbb_{\Ak} + \bX_{\Ak}\tb_{\Ak} - \bX_{A^{\circ}}\tb_{A^{\circ}} - \te}/n } \notag\\
&= \normt{\bX_{\Ak_{22}}^{\prime}\xkhB{\bX_{\Ak}\bbb_{\Ak}
- \bX_{\Ak_{22}}\tb_{\Ak_{22}}-\bX_{\Ak_2 \backslash \Ak_{22}}\tb_{\Ak_2 \backslash \Ak_{22}} - \te}/n} \notag\\
&\geq \cm(|\Ak_{22}|)\normt{\tb_{\Ak_{22}}} -\theta_{|\Ak_{22}|,T}\normt{\bbb_{\Ak}}
-\theta_{l_{k},l_{k}-|\Ak_{22}|}\normt{\tb_{\Ak_2 \backslash \Ak_{22}}} - \normt{\bX_{\Ak_{22}}\te/n} \notag\\
&\geq \cm(T)\normt{\tb_{\Ak_{22}}}
- \theta_{T,T} \normt{\bbb_{\Ak}} - \theta_{T,T}\JAk -\htt,
\end{align*}
where the first equality uses the definition of $\dkk$, the second equality uses the the definition of $\bbb$ and $y$, the third equality is simple algebra, the first inequality uses the triangle inequality, \eqref{lem2-1} and the definition of $\theta_{a,b}$,  and the last inequality uses
the monotonicity property of $\cm(\cdot)$, $\theta_{a,b}$ and the definition of $\htt$.
This proves \eqref{lem6-3}.

Let $i_{k} \in \Ak_{22}$ such that
$\normi{\tb_{\Ak_{22}}} = \abs{\tb_{i_{k}}}$.
\begin{align*}
\abs{\dkk_{i_{k}}}&= \normi{\bX_{i_{k}}^{\prime}(\bX_{\Ak}\bbb_{\Ak}
- \bX_{i_{k}}\tb_{i_{k}}-\bX_{\Ak_2 \backslash \ i_{k}}\tb_{\Ak_2 \backslash \ i_{k}} - \te)/n} \notag\\
&\geq \abs{\tb_{i_{k}}}-T\mu\normi{\bbb_{\Ak}}-l_{k}\mu\normi{\tb_{\Ak_2 \backslash \ i_{k}}}-\normi{\bX_{i_{k}}^{\prime}\te}\\
&\geq \normi{\tb_{\Ak_{22}}} - T\mu\normi{\bbb_{\Ak}} - T\mu\JAki-\hit,
\end{align*}
where the first  equality  is derived from the first three equalities in the proof of  \eqref{lem6-3} by replacing $\Ak_{22}$ with $i_{k}$,
the first inequality  uses the triangle inequality and  \eqref{lem2-5},  and the last inequality  uses the definition of $\hit$.
Then \eqref{lem6-3'} follows  by rearranging the terms in the above inequality.
This completes the proof of Lemma  \ref{lem6}.
\end{proof}

\begin{lem}\label{lem7}
\begin{align}
\normi{\bbk}\vee\normi{\dk} &= \max\{|\bbk_{i}|+|\dk_{i}|\big|i\in S\}, \forall k \geq 1. \label{lem7-1}\\
\normi{\bbkk_{\Ak_{11}}} + \normi{\dkk_{\Ak_{22}}}&\leq \abs{\bbkk_{\Ik_{33}}}_{min} \wedge \abs{\dkk_{\Ik_{44}}}_{min}.\label{lem7-2}\\
\normt{\bbkk_{\Ak_{11}}} + \normt{\dkk_{\Ak_{22}}}&\leq \sqrt{2} \xkhB{ \normt{\bbkk_{\Ik_{33}}} + \normt{\dkk_{\Ik_{44}}}}.\label{lem7-3}
\end{align}
\end{lem}

\begin{proof}
By the definition of Algorithm \ref{alg1} we have $\bbk_{i}\dk_{i} = 0$, $\forall i\in S$, $\forall k\geq 1$, i.e., \eqref{lem7-1} holds.
\eqref{lem7-2} follows from  the definition of $\Ak_{11}$, $\Ak_{22}$, $\Ik_{33}$, $\Ik_{44}$ and \eqref{lem7-1}. Now
\begin{align*}
\frac{1}{2} (\normt{\bbkk_{\Ak_{11}}} + \normt{\dkk_{\Ak_{22}}})^2
&\leq \normt{\bbkk_{\Ak_{11}}}^2 + \normt{\dkk_{\Ak_{22}}}^2\\
&\leq (\normt{\bbkk_{\Ik_{33}}} + \normt{\dkk_{\Ik_{44}}})^2,
\end{align*}
where the first inequality is simple algebra, and the second inequality uses \eqref{lemequl} and \eqref{lem7-2},
Thus (\ref{lem7-3}) follows. This completes the proof of Lemma  \ref{lem7}.
\end{proof}

\begin{lem}\label{lem8}
\begin{align}
\normt{\bbkk_{\Ik_{33}}} \leq \normt{\bbb_{\Ik_{33}}}.\label{lem8-1}
\end{align}
Furthermore, suppose (A1) holds.
We have
\begin{align}
\normi{\dkk_{\Ik_{44}}}&\leq T\mu\normi{\bbb_{\Ak}} +  T\mu\JAki +\hit.\label{lem8-2'}\\
\normt{\dkk_{\Ik_{44}}}&\leq \theta_{T,T}\norm{\bbb_{\Ak}} +  \theta_{T,T}\JAk + \htt \quad {if} \quad \bX\sim \textrm{SRC}(\cm(T), \cp(T),T). \label{lem8-2}
\end{align}
\end{lem}
\begin{proof}
By the definition of $\bbb$, the triangle inequality, and the fact that
 $\tb$ vanishes on   $\Ak \cap \Ik_{33}$, we have
\begin{align*}
\normt{\bbkk_{\Ik_{33}}} & = \normt{\bbb_{\Ik_{33}}+\tb_{\Ik_{33}}}
\leq \normt{\bbb_{\Ik_{33}}} + \normt{\tb_{\Ak \cap \Ik_{33}}} = \normt{\bbb_{\Ik_{33}}}.
\end{align*}
So \eqref{lem8-1} follows. Now
\begin{align*}
\normt{\dkk_{\Ik_{44}}}
&= \normt{\bX_{\Ik_{44}}^{\prime}\xkhB{\bX_{\Ak}\bbb_{\Ak} - \bX_{\Ak_2}\tb_{\Ak_2} - \te}/n}\notag\\
&\leq \theta_{|\Ik_{44}|,T}\normt{\bbb_{\Ak}} + \theta_{|\Ik_{44}|,l_{k}}\normt{\tb_{\Ak_2}}
 +\normt{\bX_{\Ik_{44}}^{\prime}\te} \notag\\
&\leq \theta_{T,T}\normt{\bbb_{\Ak}} +  \theta_{T,T}\JAk + \htt,
\end{align*}
where the first  equality  is derived from the first three equalities in
the proof of  \eqref{lem6-3} by replacing $\Ak_{22}$ with $\Ik_{44}$,
the first inequality uses the triangle inequality and the definition of $\theta_{a,b}$, and the last inequality uses
the monotonicity property of  $\theta_{a,b}$ and $\htt$. This implies
\eqref{lem8-2}.
Finally, \eqref{lem8-2'} can be proved similarly by using \eqref{lem2-5} and  \eqref{lem4-2}.
This completes the proof of Lemma  \ref{lem8}.
\end{proof}

\begin{lem}\label{lem9}
Suppose (A1) holds.
\begin{itemize}
\item[(i)]
If  $\bX\sim \textrm{SRC}(T, \cm(T), \cp(T))$,
then
\begin{equation}\label{lem9-1}
\JAkk\leq \gamma \JAk+ \frac{\gamma} {\theta_{T,T}} \htt,
\end{equation}
\item[(ii)]
If  $(T-1)\mu<1$, then
\begin{equation}\label{lem9-1'}
\JAkki\leq \gamma_{\mu} \JAk+ \frac{3+2\mu}{1-(T-1)\mu} \hit.
\end{equation}
\end{itemize}
\end{lem}
\begin{proof}
\begin{align*}
\JAkk&\leq \normt{\tb_{\Ak_{11}}} + \normt{\tb_{\Ak_{22}}}\\
&\leq(\normt{\bbkk_{\Ak_{11}}} + \normt{\dkk_{\Ak_{22}}} + \normt{\bbb_{\Ak_{11}}} +  \theta_{T,T} \normt{\bbb_{\Ak}}
+ \theta_{T,T}\JAk+ \htt)/ \cm(T)\\
&\leq(\sqrt{2}(\normt{\bbkk_{\Ik_{33}}} + \normt{\dkk_{\Ik_{44}}}) +   \normt{\bbb_{\Ak_{11}}}+ \theta_{T,T}\norm{\bbb_{\Ak}} +  \theta_{T,T}\JAk + \htt)/ \cm(T)\\
&\leq((2+(1+\sqrt{2})\theta_{T,T})\normt{\bbb}+ (1+\sqrt{2})\theta_{T,T}\JAk+ (1+\sqrt{2})\htt)/\cm(T)\\
&\leq ( \frac{2\theta_{T,T}+(1+\sqrt{2})\theta_{T,T}^2}{\cm(T)^2} + \frac{(1+\sqrt{2})\theta_{T,T}}{\cm(T)}) \JAk\\
&+ (\frac{2+(1+\sqrt{2})\theta_{T,T}}{\cm(T)^2}+ \frac{1+\sqrt{2}}{\cm(T)})\htt,
\end{align*}
where the first inequality is \eqref{lem6-1}, the second inequality uses  \eqref{lem6-2} and \eqref{lem6-3}, the third inequality uses \eqref{lem7-3}, the fourth inequality uses the sum of \eqref{lem8-1} and \eqref{lem8-2}, and the last inequality follows from \eqref{lem5-4}.
This implies     \eqref{lem9-1} by noticing
the definitions of $\gamma$. Now
\begin{align*}
\JAkki&\leq \normi{\tb_{\Ak_{11}}} + \normt{\tb_{\Ak_{22}}}\\
&\leq\normi{\bbkk_{\Ak_{11}}} + \normi{\dkk_{\Ak_{22}}} + \normi{\bbb_{\Ak_{11}}}+ T\mu\normi{\bbb_{\Ak}} + T\mu\JAki + \hit.\\
&\leq\normi{\dkk_{\Ik_{44}}}+ \normi{\bbb_{\Ak_{11}}}+ T\mu\normi{\bbb_{\Ak}} + T\mu\JAki + \hit\\
&\leq\normi{\bbb_{\Ak_{11}}}+ 2T\mu\normi{\bbb_{\Ak}} + 2T\mu\JAki + 2\hit\\
&\leq  (\frac{(1+2T\mu)T\mu}{1-(T-1)\mu}+2T\mu)\JAki+ \frac{3+2\mu}{1-(T-1)\mu}\hit,
\end{align*}
where the first inequality is \eqref{lem6-1'}, the second inequality uses  \eqref{lem6-2'} and \eqref{lem6-3'}, the third inequality uses \eqref{lem7-2}, the fourth inequality uses the sum of  \eqref{lem8-2'}, and the last inequality 
follows from \eqref{lem5-4'}.
Thus part (ii) of Lemma \ref{lem9} follows by noticing the definitions of $\gamma_{\mu}$.
\end{proof}


\medskip\noindent
\textbf{Proof of Theorem \ref{thmt}.}
\begin{proof}
Suppose $\gamma < 1$.  By using \eqref{lem9-1} repeatedly,
\begin{align*}
\JAkk &\leq \gamma\JAk + \frac{\gamma}{\theta_{T,T}} \htt\\
 &\leq \gamma(\gamma E_2(A^{k-1}) + \frac{\gamma}{\theta_{T,T}} \htt) + \gamma \htt\\
 & \leq \cdots \\
&\leq  \gamma^{k+1} E_2(A^0) + \frac{\gamma}{\theta_{T,T}}(1 + \gamma + \cdots + \gamma^k) \htt\\
&< \gamma^{k+1}\normt{\tb} + \frac{\gamma}{(1-\gamma)\theta_{T,T}}\htt,
\end{align*}
i.e., \eqref{thmt-1} holds. Now
\begin{align*}
\normt{\bbkk - \tb} &\leq (1+\frac{\theta_{T,T}}{\cm(T)})\JAk + \frac{\htt}{\cm(T)} \\
&\leq(1+\frac{\theta_{T,T}}{\cm(T)})\zkhB{\gamma^k\normt{\tb} + \frac{\gamma\theta_{T,T}}{1-\gamma}\htt}  \\
&= (1+\frac{\theta_{T,T}}{\cm(T)})\gamma^k\normt{\tb} + \zkhB{\frac{\gamma\theta_{T,T}}{(1-\gamma)}(1+\frac{\theta_{T,T}}{\cm(T)}) + \frac{1}{\cm(T)}}\htt,
\end{align*}
where the first inequality follows from  \eqref{lem5-1}, the second inequality uses \eqref{thmt-1}, and the third line follows after some algebra.
Thus \eqref{thmt-2} follows by noticing the definitions of $b_1$ and $b_2$. This completes the proof of part (i) of Theorem \ref{thmt}.
\eqref{thmt-3} and \eqref{thmt-4} follow from \eqref{thmt-1}, \eqref{lem4-1} and  \eqref{thmt-2}, \eqref{lem4-1} respectively.
This completes  the proof of Theorem \ref{thmt}.
\end{proof}

\medskip\noindent
\textbf{Proof of Corollary \ref{cort}.}
\begin{proof}
By \eqref{thmt-2},
  \begin{align*}
\normt{\bbkk-\tb} &\leq b_1\gamma_1^k\normt{\tb} + b_2\htt\\
&\leq b_1 \htt  + b_2\htt \quad \textrm{if} \quad 
k\geq \log_{\frac{1}{\gamma}} \frac{\sqrt{J}\bar{M}}{\htt}
\end{align*}
where 
the second inequality follows after  some algebra.
By \eqref{thmt-1},
\begin{align*}
\normt{\tb|_{A^{*}_{J} \backslash \Ak}} &\leq \gamma^{k}\normt{\tb} + \frac{\gamma \theta_{T,T}}{1-\gamma}\htt\\
&\leq \gamma^{k}\sqrt{J}\bar{M} +\xi \bar{m} \\
&< \bar{m}\quad \textrm{if} \quad k\geq \log_{\frac{1}{\gamma}} \frac{\sqrt{J}R}{1-\xi},
\end{align*}
where 
the second inequality uses the assumption $ \bar{m} \geq \frac{\gamma\htt}{(1-\gamma)\theta_{T,T}\xi} $ with $0 < \xi < 1$, and the last inequality follows after some simple algebra.
This implies  $A^{*}_{J} \subset \Ak$ if $k\geq \log_{\frac{1}{\gamma}} \frac{\sqrt{J}R}{1-\xi}$.
This proves part (i).
The proof of part (ii) of is similar to that of part (i) by using \eqref{lem4-1}, we omit it here.
Suppose $\beta^*$ is exactly  $K$-sparse and $T=K$ in SDAR. It follows from  part (ii) that with probability at least $1-2\alpha$, $A^{*} = \Ak$ if $k\geq \log_{\frac{1}{\gamma}} \frac{\sqrt{K}R}{1-\xi}$. Then part (iii) holds by   showing  that $\Akk = A^{*}  $. Indeed, by \eqref{lem9-1} and \eqref{lem4-1} we have
\begin{align*}
\normt{\tb|_{A^{*}\backslash \Akk}}&\leq \gamma \normt{\tb|_{A^{*} \backslash \Ak}} + \frac{\gamma}{\theta_{T,T}} \s\sqrt{K}\sqrt{2\log(p/\alpha)/n}\\
&=\frac{\gamma}{\theta_{T,T}} \s\sqrt{K}\sqrt{2\log(p/\alpha)/n}.
\end{align*}
Then $\Akk = A^{*}  $ by using  the assumption that $m \geq \frac{\gamma}{(1-\gamma)\theta_{T,T}\xi}\s\sqrt{K}\sqrt{2\log(p/\alpha)/n}>\frac{\gamma}{\theta_{T,T}}\s\sqrt{K}\sqrt{2\log(p/\alpha)/n}$.
This completes the proof of Corollary  \ref{cort}.
\end{proof}

\medskip\noindent
\textbf{Proof of Theorem \ref{thmi}.}
\begin{proof}
 Suppose $T\mu\leq 1/4$, some algebra shows $\gamma_{\mu}<1$ and $\frac{1+\mu}{1-(T-1)\mu}<\frac{3+2\mu}{1-(T-1)\mu}<4$. Now Theorem \ref{thmi} can be proved in a way similar to Theorem \ref{thmt} by using \eqref{lem9-1'}, \eqref{lem4-2} and \eqref{lem5-1'}. We omit it here.
This completes the proof of Theorem \ref{thmi}.
\end{proof}

\medskip\noindent
\textbf{Proof of Corollary \ref{cori}.}
\begin{proof}
The proofs of part (i) and part (ii) are similar to that of Corollary \ref{cort},
we omit them here.
Suppose $\beta^*$ is exactly $K$-sparse and $T=K$ in SDAR. It follows from  part (ii) that with probability at least $1-2\alpha$, $A^{*} = \Ak$ if $k\geq \log_{\frac{1}{\gamma_{\mu}}} \frac{R}{1-\xi}$. Then part (iii) holds by   showing  that $\Akk = A^{*} $. Indeed, by \eqref{lem9-1'}, \eqref{lem4-2} and $\frac{3+2\mu}{1-(T-1)\mu}<4$ we have
\begin{align*}
\normi{\tb|_{A^{*}\backslash \Akk}}&\leq \gamma_{\mu} \normi{\tb|_{A^{*} \backslash \Ak}} + 4 \s\sqrt{2\log(p/\alpha)/n}\\
&=4 \s\sqrt{2\log(p/\alpha)/n}.
\end{align*}
Then $\Akk = A^{*}  $ by using  the assumption that $$m\geq \frac{4}{\xi (1-\gamma_{\mu})}\s\sqrt{2\log(p/\alpha)/n}>4\s\sqrt{2\log(p/\alpha)/n}.$$
This completes the proof of Corollary \ref{cori}.
\end{proof}

%


\end{document}